\documentclass[10pt]{article}

\usepackage[T1]{fontenc}
\usepackage{lmodern}
\usepackage[frenchmath]{mathastext}							% sets default math font to the default text font (in italics)										% improves spacing
\usepackage{natbib}       									
\usepackage{amsmath}                        % need for subequations
\numberwithin{equation}{section}
\usepackage{graphicx} 											% if you compile LaTeX => PS => PDF, then uncomment this and uncomment next two lines
\usepackage{adjustbox,float}
\usepackage{enumitem}												% allows you to customize spacing of lists
\usepackage{mdwlist}												% enumerate* itemize* are compact
\usepackage[dvipsnames]{xcolor}							% gives you the colors you will use in the next line
\usepackage[plainpages=false, pdfpagelabels]{hyperref} % enables and customizes hyperlinks
	\hypersetup{
		colorlinks   = true,
		citecolor    = RoyalBlue, %NavyBlue,
		linkcolor    = RubineRed, %Maroon,
		urlcolor     = Turquoise
	}
\usepackage{amssymb}                        % gives you \mathbb{} font
\usepackage[mathscr]{eucal}                 % gives you \mathscr font
\usepackage{dsfont}													% gives you \mathds{} font
\usepackage[paperwidth=8.5in,paperheight=11in,top=1.25in, bottom=1.25in, left=1.0in, right=1.0in]{geometry} %use with 10pt font
\usepackage{mathtools}                      % need for `show only references'
\mathtoolsset{showonlyrefs=true}          % only equations which are labeled AND referenced will numbered use \eqref{} instead of (\ref{})
\usepackage{amsthm}                         % need for theorem-proof environment
	\allowdisplaybreaks                       % allows page breaks for long equations...you can prevent a page-break with \\*
	\theoremstyle{plain}
	\newtheorem{theorem}{Theorem}
	\numberwithin{theorem}{section}
	\newtheorem{lemma}[theorem]{Lemma}       	% [theorem] ==> theorems and lemmas will share a counter
	\newtheorem{proposition}[theorem]{Proposition}
	\newtheorem{corollary}[theorem]{Corollary}
	\theoremstyle{definition}

	\newtheorem{remark}[theorem]{Remark}

\newcommand{\diag}{\operatorname{diag}}

\newcommand{\tr}{\operatorname{tr}}

\newcommand \vp {v}

\newcommand \VF {V_F}
\newcommand \AF {A}
\newcommand \BF {B}
\newcommand \CE{\operatorname{CE}}

\newcommand\Eb{\mathds{E}}

\newcommand\Pb{\mathds{P}}

\newcommand\Rb{\mathds{R}}

\newcommand\Ac{\mathscr{A}}

\newcommand\Fc{\mathscr{F}}

\newcommand\Lc{\mathscr{L}}

\newcommand\Jc{\mathscr{J}}

\newcommand\sig{\sigma}
\newcommand\Sig{\Sigma}

\newcommand\gam{\gamma}

\newcommand\kap{\kappa}
\newcommand\tht{\theta}
\newcommand\0{\textbf{0}}

\newcommand\Fv{\textbf{F}} %<--- with concrete fonts there is no bolded math font so you changed to bold text font
\newcommand\Sv{\textbf{S}}
\newcommand\Zv{\textbf{Z}} %<--- with concrete fonts there is no bolded math font so you changed to bold text font
\newcommand\av{\mathbf{a}}

\newcommand\ev{\mathbf{e}}

\newcommand\gv{\mathbf{g}}
\newcommand\rv{\mathbf{r}}

\newcommand\zv{\mathbf{z}}

\newcommand\mv{\mathbf{m}}

\newcommand\Wv{\textbf{W}}

\newcommand\etav{{\boldsymbol\eta}}

\newcommand\muv{{\boldsymbol\mu}}
\newcommand\piv{{\boldsymbol\pi}}
\newcommand\thtv{{\boldsymbol\tht}}

\newcommand\Thtv{{\boldsymbol\Theta}}

\newcommand\Ct{\widetilde{C}}

\newcommand\Ht{\widetilde{H}}

\newcommand\Zvt{\widetilde{\textbf{Z}}}

\newcommand\muvt{\widetilde{\muv}}

\newcommand\Sigt{\widetilde{\Sig}}
\newcommand\sigt{\widetilde{\sig}}

\newcommand\dd{\mathrm{d}}
\newcommand\ee{\mathrm{e}}

\newcommand\Cov{\operatorname{Cov}}

\providecommand{\keywords}[1]{\textbf{\textit{Keywords: }} #1}

\begin{document}

\title{Optimal Trading of a Basket of Futures Contracts}

\author{
Bahman Angoshtari
\thanks{Department of Applied Mathematics, University of Washington, Seattle WA 98195.  \textbf{e-mail}: \url{bahmang@uw.edu}}
\and
Tim Leung
\thanks{Department of Applied Mathematics, University of Washington, Seattle WA 98195.  \textbf{e-mail}: \url{timleung@uw.edu}}
}

\date{This version: \today}

\maketitle

\begin{abstract}
We study the problem of dynamically trading   multiple futures contracts with different underlying assets. To capture the joint dynamics of stochastic bases for all traded futures, we propose a new model involving a multi-dimensional scaled Brownian bridge that is stopped before price convergence. This leads to the analysis of the corresponding Hamilton-Jacobi-Bellman (HJB) equations, whose solutions are derived in semi-explicit form. The resulting optimal trading strategy is a long-short policy that accounts for whether the futures are in contango or backwardation. Our model also allows us to quantify and compare the values of trading in the futures markets when the underlying assets are traded or not. Numerical examples are provided to illustrate the optimal strategies and  the effects of   model parameters.
\end{abstract}

\keywords{  futures, stochastic basis,  Brownian bridge, utility maximization} 

% \emph{\textbf{JEL Classification:}}  C61 G11 G13

%-----------------------------------------------------------------------------------
%
%       SECTION: 		Introduction
%
%-----------------------------------------------------------------------------------

\section{Introduction}\label{sec:Intro}

%Futures are bilateral contracts of agreement to buy or sell an asset at a pre-determined price at a pre-specified time in the future. The underlying (spot) asset can be a physical commodity, market index, or financial instrument. 
%
%Futures are standardized and exchange-traded. 

Following the \emph{financialization} of commodity markets in the early 2000s, commodity futures have gained popularity among fund managers and institutional investors. It is estimated that the net institutional investors' holding in various commodity futures indices has risen from \$15 billion in 2003 to more than \$200 billion in 2008.\footnote{See CFTC Press Release 5542-08: \url{https://www.cftc.gov/PressRoom/PressReleases/pr5542-08}.} The Chicago Mercantile Exchange (CME), which is the world's largest futures exchange, averages over 15 million futures contracts traded per day.\footnote{According to CME group report: \url{https://www.cmegroup.com/daily_bulletin/monthly_volume/Web_ADV_Report_CMEG.pdf}} Within the universe of hedge funds and alternative investments, futures funds play an integral role with hundreds of billions under management. This motivates us to study the problem of trading futures.

%That this deregulation has benefited the commodity markets as a whole, however, is still subject to debate. Proponents insist on benefits such as added liquidity and ease in hedging pressure. Critics argue that the deregulation has led to increase in commodity price, volatility, and comovements as well as market anomalies such as non-convergence of futures and spot prices at the delivery date.  See \cite{vanHuellen2018} and the references therein for further discussion.

There is substantial theoretical and empirical evidence pointing to long-short trading strategies for commodity futures. On the empirical side, various studies have highlighted that a multitude of long-short trading strategies provide superior performance relative to long-only strategies. See \cite{Miffre2016} for a survey. On the theoretical side, there are two long standing asset pricing theories, with substantial empirical support, suggesting long-short trading strategies for speculative commodity futures. 
% that suggest the following trading strategy for a speculator trading a commodity futures. Short the futures when the asset is ``\emph{contangoed}'' (meaning that the futures curve is upward sloping), and long the futures when the asset is ``\emph{backwardated}'' (meaning that the futures curve is downward sloping).

According to \emph{the theory of storage},\footnote{See \cite{Kaldor1939}, \cite{Working1949}, and \cite{Brennan1958}.} when a commodity has an abundant supply, its market  is in \emph{contango},  meaning that the futures prices are expected to fall over time (as time-to-delivery decreases). In this situation, a speculators may seek to take a short position in the futures. On the other hand, if the market is in \emph{backwardation},  then the futures prices tend to \mbox{decrease} over time. In this case, the benefit of holding inventory (referred to as the \emph{convenience yield}) is higher than financing and storage costs. Since inventory holders benefit from holding the commodity, the futures price is expected to rise with time, leading speculators to long the futures. The second pricing theory, the so-called \emph{hedging pressure hypothesis,}\footnote{See \cite{Cootner1960}.}connect the phenomena of contango and backwardation to the  risk-premium paid by hedgers to speculators.  Intuitively, hedgers who are long in the commodity or its futures  need to attract speculators to take short positions as their counterparties. To this end, hedgers are willing to accept higher future prices (which benefits a speculator that enters a short position), resulting in a contango. Similarly, a net short position for the hedgers results in a demand for long speculators and results in a backwardated market. This argument highlights that contangoed (resp. backwardated) commodity markets pay risk premia for short (resp. long) positions.

These theories suggest that a speculator should take a short position in a futures when the commodity is contangoed and a long position when it is backwardated. The \emph{basis} of a futures contract, defined herein as the ratio of the futures price to the forward price, is a signal for determining if, and to what degree, a commodity is contangoed or backwardated. As such, a speculative trading strategy involving a futures contract should be driven by its basis. To formulate an optimal trading problem involving futures contracts, it is thus imperative to incorporate the dynamics of the basis.

Early studies of optimal futures trading that incorporates the dynamics of basis include \cite{BrennanSchwartz1988} and \cite{BrennanSchwartz1990}. They assumed that the basis of an index futures follows a scaled Brownian bridge and calculated the value of the embedded timing options to trade the basis. They then used the option prices to devise open-hold-close strategies involving the index futures and the underlying index. Also under a Brownian bridge model,   \cite{dai2011optimal} provided an alternative trading strategy and specification of transaction costs. Another related work by \cite{LiuLongstaff2004} assumed that the basis follows a scaled Brownian bridge and the investor is subject to a collateral constraint. They derived the  closed-form   strategy that maximizes  the expected logarithmic utility of terminal wealth.

In the aforementioned studies, however, the market model contains arbitrage. Indeed, it is assumed that the basis, which is a tradable asset, converges to zero at a fixed future time. Instead of specifying the dynamics of the basis, one could start with a model for the commodity price and then derive the no arbitrage futures price using risk-neutral pricing theory. One popular class of such models are the so-called \emph{convenience yield models}, introduced by \cite{GibsonSchwartz1990} who proposed a factor model in which the drift of the spot price was driven by a \emph{convenience yield} process assumed to be an Ornstein-Uhlenbeck process. This model was subsequently extended to incorporate stochastic interest rate and jumps, see \cite{Schwartz1997}, \cite{HilliardReis1998}, and
\cite{CarmonaLudkovski2004}. 

In this paper, we propose a model to capture the joint dynamics of stochastic bases using a multi-dimensional Brownian bridges that are stopped before convergence. We then formulate and solve two stochastic optimal trading problems. In one setting, the underlying assets are not traded, and the trading strategies involve only futures. In the other, both the futures and   underlying assets are traded. In both cases, we derive the optimal trading strategies by solving the corresponding Hamilton-Jacobi-Bellman (HJB) equations. The trading strategies are expressed  semi-explicitly in terms of the solution of a matrix Riccati differential equation, and are numerically illustrated under a number of different trading scenarios.

%A more general approach is to apply the Heath-Jarrow-Morton framework, originally introduced for pricing of fixed-income securities, to model the forward curve of futures contracts. See \cite{MiltersenSchwartz1998} and \cite{BjorkLanden2002}, among others. 
%\cite{LeungLiLiZheng2015} study the problem of trading futures with transaction costs when the underlying spot price is mean-reverting.

%\red{Tim: Please add papers that you see fit, these are continuous time optimal trading problems involving futures}.

Among the authors' recent related studies, \cite{LeungYan2018} and \cite{LeungYan2019} applied utility maximization approach to derive dynamic pairs trading strategies for futures under two-factor spot models.  Most recently  and relevantly, \cite{AngoshtariLeung2019} analyzed  the problem of dynamically trading a futures contract and its underlying asset. The associated basis is modeled by a Brownian bridge, but the process is stopped early to capture the non-convergence of prices at the end of trading horizon. In the current paper, we extend this approach in two directions. Firstly, we model the joint dynamics of stochastic bases using a multi-dimensional Brownian bridges. Secondly, multiple futures and spot assets are traded.  In addition, we also include the case where the only the futures  are traded but the spot assets are not.

The rest of the paper is organized as follows. We introduce our market model in Section \ref{sec:model} and analyze its properties in Section \ref{sec:Properties}. In Section \ref{sec:Futures}, we formulate and solve the problem of optimally trading a portfolio of futures contracts without  the underlying assets. In Section \ref{sec:SpotFutures}, we derive the optimal strategy when  futures and underlying assets are all traded. We provide a series of illustrative numerical examples in Section \ref{sec:example}. Concluding remarks are included in Section \ref{sec:conclude}, and  lengthy proofs are   in the Appendix.

%-----------------------------------------------------------------------------------
%
%       SECTION: 		Market Model
%
%-----------------------------------------------------------------------------------
\section{Market setting}\label{sec:model}
Consider a market with a riskless asset that pays interest at a constant rate $r\ge0$, $N$ non-dividend paying assets $S_{1},\ldots, S_{N}$, and $N$ futures contracts $F_1,\ldots,F_N$ on the assets (that is, one futures contract per each asset). We assume that the expiry date of the $i$-th futures contract is  $T_i$ and consider trading in this market over the horizon $[0,T]$ where $0<T<T_i$ for all $i\in\{1,\dots,N\}$. In particular, trading stops before the expiry of the futures contracts.

Since the interest rate is assumed to be deterministic, the futures and forward prices must coincide and we must have $F_{t,i} = S_{t,i}\ee^{r(T_i-t)}$. In practice, however, futures and forward prices are different because of various market uncertainties and inefficiencies. Indeed, unexpected fluctuations in the supply of the underlying asset or changes in the net position of the hedgers in the market can push the futures price up or down. See the discussion of the \emph{theory of storage} and the \emph{hedging pressure hypothesis} in Section \ref{sec:Intro}.  In order to capture real markets more closely, we assume that the \emph{bases}
\begin{align}
	\frac{F_{t,i}}{S_{t,i}}\ee^{-r(T_i-t)};\quad 0\le t\le T, i\in\{1,\ldots,N\},
\end{align}
are stochastic processes and we propose a mathematical model for them.

The log-value of the random basis for the $i$-the futures contract is defined by 
\begin{align}\label{eq:Z}
	Z_{t,i} := \log\left(\frac{F_{t,i}}{S_{t,i}}\right) - r (T_i - t);\quad i\in\{1,\ldots,N\}.
\end{align}
We assume that the futures price processes $(F_{t,i})_{0\le t\le T}$ satisfy
\begin{align}\label{eq:Fi}
	\dd F_{t,i} &= F_{t,i}\left[\left(\mu_{i,F} + \frac{\eta_{i,F}}{T_i-t} Z_{t,i}\right)\dd t + \sum_{j=1}^i \sigt_{i,j} \dd W_{t,j}\right],
\intertext{and the asset price $(S_{t,i})_{0\le t\le T}$ satisfy}
	\label{eq:Si}
	\dd S_{t,i} &= S_{t,i}\left[\left(\mu_{i,S} + \frac{\eta_{i,S}}{T_i-t} Z_{t,i}\right)\dd t + \sum_{j=1}^{N+i}\sigt_{N+i,j} \dd W_{t,j}\right].
	% \dd S_{t,i} &= S_{t,i}\left(\mu_i \dd t + \sum_{j=1}^{N+i}\sigt_{N+i,j} \dd W_{t,j}\right).
\end{align}
Here, $\big(\Wv_t^\top=(W_{t,1},\dots,W_{t,2N})\big)_{0\le t\le T}$ is a standard $2N$ dimensional Brownian motion in a filtered probability space $\big(\Omega,$ $\Fc,$ $\Pb,$ $(\Fc_t)_{t\ge0}\big)$, where $(\Fc_t)_{t\ge0}$ is generated by the Brownian motion and satisfies the usual conditions. It is assumed that $\mu_{i,S}$, $\mu_{i,F}$, $\eta_{i,F}$ and $\eta_{i,S}$ are constants such that $\eta_{i,F}<\eta_{i,S}$. Recall, also, our standing assumption that $T<T_i$.

Next, we discuss the correlation between the assets as well as clarifying the role of the parameters $\sigt_{i,j}$. Define the (instantaneous) \emph{covariance matrix}
\begin{align}\label{eq:Cov}
	\Sig =
	\begin{pmatrix}
		\Sig_\Fv & \Sig_{\Fv\Sv}\\
		\Sig_{\Fv\Sv}^\top & \Sig_{\Sv}
	\end{pmatrix}
	:= \Sigt\Sigt^\top,
\end{align}
in which $\Sigt$ is a $2N\times 2N$ upper triangular block matrix
\begin{align}\label{eq:Sigt}
	\Sigt =
	\begin{bmatrix}
		\Sigt_\Fv & \0\\
		\Sigt_{\Sv\Fv} & \Sigt_\Sv
	\end{bmatrix}
	:=\begin{bmatrix}
		\sigt_{1,1} &\dots & 0 & 0 &\dots& 0\\
		\vdots &\ddots &\vdots & \vdots &\ddots& \vdots\\
		\sigt_{N,1} &\dots &\sigt_{N,N} & 0 &\dots& 0\\
		\sigt_{N+1,1} &\dots &\sigt_{N+1,N} & \sigt_{N+1,N+1} &\dots& 0\\
		\vdots &\ddots &\vdots & \vdots &\ddots& \vdots\\
		\sigt_{2N,1} &\dots &\sigt_{2N,N} & \sigt_{2N,N+1} &\dots& \sigt_{2N,2N}
	\end{bmatrix}.
\end{align}
We assume that $\Sig$ is positive definite. Let $\sig_{i,j}$ denote the $ij$-th element of $\Sig$. For $i,j\in\{1,\dots,N\}$, we have:
\begin{enumerate}
	\item[(i)] $\sig_{i,j}$ is the instantaneous covariance between the $i$-th and $j$-th futures. In particular, $\sig_{i,F}:=\sqrt{\sig_{i,i}}$, $i\in\{1,\dots,N\}$, is the volatility of the $i$-th future.
	
	\item[(ii)] $\sig_{N+i,N+j}$ is the instantaneous covariance between the $i$-th and $j$-th underlying assets. In particular, $\sig_{i,S}:=\sqrt{\sig_{N+i,N+i}}$, $i\in\{1,\dots,N\}$, is the volatility of the $i$-th underlying asset.
	
	\item[(iii)] $\sig_{N+i,j}=\sig_{j,N+i}$ is the instantaneous covariance between the $i$-th futures and the $j$-th underlying assets.
\end{enumerate}

We end this section by providing a matrix notation which will be used for simplifying later arguments. We define
\begin{align}
	\Sv_t &:= (S_{t,1},\ldots,S_{t,N})^\top,\\*
	\Fv_t &:= (F_{t,1},\ldots,F_{t,N})^\top,\\*
	\Zv_t &:= (Z_{t,1},\ldots,Z_{t,N}),^\top \\*
	\Wv_{t,1} &:= (W_{t,1}, \ldots, W_{t,N})^\top,\\*
	\Wv_{t,2} &:= (W_{t,N+1},\ldots,W_{t,2N})^\top,\\*
	\muv_\Sv &:= (\mu_{1,S},\dots,\mu_{N,S})^\top,\\*
	\muv_\Fv &:=(\mu_{1,F},\dots,\mu_{N,F})^\top,
\intertext{and}
	\etav_k(t) &:= \diag\left(\frac{\eta_{1,k}}{T_1-t}, \dots,\frac{\eta_{N,k}}{T_N-t}\right),	
\end{align}
for $k\in\{S,F\}$, in which $\diag(\av)$ denotes a diagonal matrix with diagonal $\av$. Under these notations, \eqref{eq:Fi} and \eqref{eq:Si} become
\begin{align}
\begin{cases}
	\dd \Fv_t = \diag(\Fv_t)\left[\big(\muv_\Fv + \etav_\Fv(t) \Zv_t\big)\dd t + \Sigt_\Fv \dd \Wv_{t,1}\right],\vspace{1em}\\
	\dd \Sv_t = \diag(\Sv_t)\left[\big(\muv_\Sv + \etav_\Sv(t) \Zv_t\big)\dd t + \Sigt_{\Sv\Fv} \dd \Wv_{t,1} + \Sigt_\Sv \dd \Wv_{t,2}\right].
\end{cases}
\end{align}
% in which we have defined
% \begin{align}
% 	\muv =
% 	\begin{pmatrix}
% 		\mu_1\\
% 		\vdots\\
% 		\mu_N
% 	\end{pmatrix},\quad
% 	\alv =
% 	\begin{pmatrix}
% 		\al_1\\
% 		\vdots\\
% 		\al_N
% 	\end{pmatrix},\quad
% \text{and}\quad
% 	K(t) &= \diag\left(\frac{\kap_1}{T+\eps_1-t}, \dots,\frac{\kap_N}{T+\eps_N-t}\right).
% \end{align}
%
% \begin{remark}
% 	By \eqref{eq:Cov}, the matrices $\Sigt_\Sv, \Sigt_\Wv,$ and $\Sigt_\Bv$ are obtained by Cholesky decomposition of the covariance matrix of the stocks and futures.\qed
% \end{remark}

\section{Properties of the stochastic basis process}\label{sec:Properties}
In this section, we provide the main properties of the model presented in the previous section. We also provide an exact discretization scheme which can be used for simulating and calibrating the model.

As stated in the following lemma, the stochastic basis $(Z_{t,i})_{0\le t\le T}$ of \eqref{eq:Z} is a scaled Brownian bridge that converges to zero at $T_i$. Therefore, the model implies that $S_i$ and $F_i$ converge at $T_i$, i.e. $\lim_{t\to T_i}(S_t / F_{t,i}) = 1$, $\Pb$-almost surely. However, since trading stops at $T$ and $T<T_i$, such a convergence is not realized in the market. As stated before, this non-convergence has practical relevance since speculative futures trades are always closed out before the delivery data.

Furthermore, Lemma \ref{lem:Bb} highlights that the stochastic basis $(Z_{t,i})_{0\le t\le T}$ is a mean-reverting process that can take positive or negative values. Note that large positive values of $Z_t$ signal a strongly \emph{contangoed} market, while large negative values indicate a strongly \emph{backwardated} market.\footnote{A commodity is said to be \emph{contangoed} if its forward curve (which is the plot of its futures prices against time-to-delivery) is increasing. The commodity is \emph{backwardated} if its forward curve is decreasing.} Therefore, the market model of Section \ref{sec:model} assumes that each futures price has a natural state, which can be a contangoed or backwardated curve, while allowing for temporary deviations from this state (say, caused by fluctuation in the inventory levels or hedging demand). Recall from our earlier discussion in Section \ref{sec:Intro}, that we expect  a trading strategy to take short futures positions when the asset is contangoed, and long positions when it is backwardated. As we will see in Sections \ref{sec:Futures} and \ref{sec:SpotFutures}, our optimal trading strategies are consistent which these expectations. 

In the statement of the lemma, we use the following notations:
\begin{align}
	\Sigt_\Zv &:= [I_N, -I_N]\Sigt,\\*
	\mv^\top &:= (m_1,\ldots,m_N),
\intertext{and}
	K(t)&:=\diag\left(\frac{\kap_1}{T_1-t},\dots,\frac{\kap_N}{T_N-t}\right),
\end{align}
in which $m_i := r + \mu_{i,F} - \mu_{i,S} - \frac{1}{2} (\sig_{i,F}^2 - \sig_{i,S}^2)$, and $\kap_i :=\eta_{i,S}-\eta_{i,F}>0$. Recall, also, that $\sig_{i,F}$ and $\sig_{i,S}$ are the volatility of the $i-$th futures and the $i$-th underlying asset, respectively.

\begin{lemma}\label{lem:Bb}
	For $0\le t\le T$, we have
	\begin{align}\label{eq:Z-SDE}
		\dd \Zv_t &= \big(\mv - K(t)\,\Zv_t\big)\dd t + \Sigt_{\Zv}\,\dd \Wv_t.
		% &= \big(\mv - K(t)\,\Zv_t\big)\dd t + \left(\Sigt_F-\Sigt_{SF}\right)\dd \Wv_{t,1} - \Sigt_S\,\dd \Wv_{t,2}.
	\end{align}
	% in which $\Sigt_\Zv = [I_N, -I_N]\Sigt$, $\mv^\top = (m_1,\ldots,m_N)$, and $K(t)=\diag\left(\frac{\kap_1}{T+\eps_1-t},\dots,\frac{\kap_N}{T+\eps_N-t}\right)$ with $m_i := r + \al_i - \mu_i - \frac{1}{2} (\sig_{F,i}^2 - \sig_{S,i}^2)$ and $\kap_i:=\eta_{i,2}-\eta_{i,1}>0$. Recall that $\sig_{F,i}$ and $\sig_{S,i}$ are the volatility of the $i-$th futures and the $i$-th asset, respectively.
%	% \begin{align}
% 	% 	m_i := \al_i - \mu_i + r - \frac{1}{2} (\sig_{F,i}^2 - \sig_{S,i}^2),
% 	% \end{align}
% 	and
% 	\begin{align}
% 		K(t) := E_\Fv(t)-E_\Sv(t)
% 		= \diag\left(-\frac{\kap_1}{T-t+\eps_1}, \dots,-\frac{\kap_N}{T-t+\eps_N}\right).
% 	\end{align}
% 	in which $\kap_i := \eta_{S,i}-\eta_{F,i}>0$.
	In particular, $(Z_{t,i})_{0\le t\le T}$ is a stopped Brownian bridge satisfying
	\begin{align}\label{eq:Bb}
		\dd Z_{t,i} = \left(m_i - \frac{\kap_i Z_{t,i}}{T_i-t}\right) \dd t + \sig_{Z,i} \dd B_{t,i},
	\end{align}
	for a constant $\sig_{Z,i}>0$ and a (one dimensional) standard Brownian motion $(B_{t,i})_{0\le t\le T}$. If we consider the solution of \eqref{eq:Bb} over $[0,T_i]$, then $Z_{T_i, i} = 0$, $\Pb$-almost surely.
\end{lemma}

\begin{proof}
	 \eqref{eq:Z-SDE} and \eqref{eq:Bb} directly follow from \eqref{eq:Z}--\eqref{eq:Si} by applying It\^o's lemma. The rest of the proof is similar to the proof of Lemma 1 in \cite{AngoshtariLeung2019} and is, thus, omitted.
\end{proof}

\begin{remark}\label{rem:arbitrage}
	In a limiting case of our model where $T_i\to T^+$, the spot and future prices converge at $T$. In this limiting case, the market model admits arbitrage, which is shown by the well-known argument for pricing forward contracts. Throughout the paper, we consider an arbitrage free model by assuming $T_i>T$. In particular, under this assumption, an argument similar to the proof of Proposition 1 of \cite{AngoshtariLeung2019} shows that the market model has a risk-neutral measure and, thus, is arbitrage free.
\end{remark}

The stochastic differential equation \eqref{eq:Z-SDE} is linear and its unique solution is given by
	\begin{align}\label{eq:Zsol}
		\diag\left(\left(\frac{T_i}{T_i-t}\right)^{\kap_i}\right)\Zv_t = \Zv_0 
		&{}+ \int_0^t \diag\left(\left(\frac{T_i}{T_i-s}\right)^{\kap_i}\right) \mv\, \dd s\\*
		&{}+ \int_0^t \diag\left(\left(\frac{T_i}{T_i-s}\right)^{\kap_i}\right) \Sigt_\Zv \dd \Wv_s;\quad 0\le t\le T.
	\end{align}
See, for example, eq. (6.6) on page 354 of \cite{KartzasShreve1991}. Here, we used the shorthand notation $\diag(a_i)=\diag\big((a_1,\dots,a_N)\big)$.

For the rest of this section, we derive a discretization scheme for the market model \eqref{eq:Z}--\eqref{eq:Si}. In Section \ref{sec:example}, we will use simulated prices based on this scheme to illustrate our results.

The following result provides an exact discretization for $(Z_t)_{0\le t\le T}$. 
   
\begin{proposition}\label{prop:Z-Sim}
	Let $(\Zv_t)_{0\le t\le T}$ be as in \eqref{eq:Zsol} and $M$ be a positive integer. We then have
	\begin{align}\label{eq:Z-AR}
		&\diag\left(\left(\frac{T_i}{T_i-\frac{n\,T}{M}}\right)^{\kap_i}\right)\Zv_{\frac{nT}{M}}\\*
		&\qquad{}=\diag \left( \left( \frac{T_i}{T_i-\frac{(n-1)\,T}{M}} \right)^{\kap_i} \right) \Zv_{\frac{(n-1)T}{M}} 
		+ \diag(\phi_{n,i})\mv + \ev_{n,\Zv},\hspace{1em}
	\end{align}
	for $n\in\{1,\dots,M\}$, in which
	\begin{align}\label{eq:phi}
		\phi_{n,i} := \begin{cases}
		\frac{(T_i)^{\kap_i}}{\kap_i-1}\left(\left(T_i-\frac{nT}{M}\right)^{1-\kap_i}
		- \left(T_i-\frac{(n-1)T}{M}\right)^{1-\kap_i}\right);
		&\quad \kap_i\in(0,1)\cup(1,+\infty),\vspace{1ex}\\
		T_i \ln\left(\frac{T_i - \frac{(n-1)T}{M}}{T_i - \frac{nT}{M}}\right);
		&\quad \kap_i = 1,
	\end{cases} 
	\end{align}
	for $i\in\{1,\dots, N\}$, and $\ev_{1,\Zv},\dots,\ev_{M,\Zv}$ are independent Gaussian random variables with zero mean and covariance matrix $\Eb(\ev_{n,\Zv}\ev_{n,\Zv}^\top)=[\eta_{ij;n}]_{N\times N}$ with
\begin{align}\label{eq:eta}
	\eta_{ij;n}:=
	\beta_{ij}\int_{\frac{(n-1)T}{M}}^{\frac{n T}{M}}
	\left(\frac{T_i}{T_i-s}\right)^{\kap_i}\left(\frac{T_j}{T_j-s}\right)^{\kap_j}
	\dd s.
\end{align}
Here, $\beta_{ij}$ is the $ij$-th element of the matrix
\begin{align}
	\Sig_\Zv := \Sigt_\Zv\Sigt_\Zv^\top=\Sig_\Fv+\Sig_\Sv-\Sig_{\Fv\Sv}-\Sig_{\Fv\Sv}^\top.
\end{align}
\end{proposition}

\begin{proof}
	% Consider sampling $(\Zv_t)_{0\le t\le T}$ at $M+1$ instances $0, T/M, 2T/M, \dots, T$.
	Define the discrete sample $\{\Zvt_n\}_{n=0}^M$ by
	\begin{align}
		\Zvt_n = \diag\left(\left(\frac{T_i}{T_i-\frac{n\,T}{M}}\right)^{\kap_i}\right)\Zv_{\frac{n\,T}{M}};\quad n\in\{0,\dots,M\}.
	\end{align}
	From \eqref{eq:Zsol}, we have
	\begin{align}\label{eq:Zvt}
		\Zvt_n = \Zvt_{n-1} &{}+ \int_{\frac{(n-1)T}{M}}^{\frac{n T}{M}} \diag\left(\left(\frac{T_i}{T_i-s}\right)^{\kap_i}\right) \mv\, \dd s\\*
		&{}+ \int_{\frac{(n-1)T}{M}}^{\frac{n T}{M}} \diag\left(\left(\frac{T_i}{T_i-s}\right)^{\kap_i}\right) \Sigt_\Zv \dd \Wv_s.
	\end{align}
	Direct calculation shows that,
	\begin{align}
		&\int_{\frac{(n-1)T}{M}}^{\frac{n T}{M}} \diag\left(\left(\frac{T_i}{T_i-s}\right)^{\kap_i}\right) \mv\, \dd s = (\phi_{n,1}\dots,\phi_{n,N})^\top,
	\end{align}
	with $\phi_{n,i}$ given by \eqref{eq:phi}.
	Furthermore, let
	\begin{align}\label{eq:eZ}
		\ev_{n,\Zv} := \int_{\frac{(n-1)T}{M}}^{\frac{n T}{M}} \diag\left(\left(\frac{T_i}{T_i-s}\right)^{\kap_i}\right) \Sigt_\Zv \dd \Wv_s;\quad n\in\{1,\dots,M\}.
	\end{align}
	Then, $\ev_{1,\Zv}, \dots, \ev_{N,\Zv}$ are independent Gaussian random variables with mean $\Eb(\ev_{n,\Zv})=\0_{N\times1}$ and covariance matrix $\Eb(\ev_{n,\Zv}\ev_{n,\Zv}^\top)=[\eta_{ij;n}]_{N\times N}$  with $\eta_{ij;n}$ given by \eqref{eq:eta}. The result follows from replacing the integrals in \eqref{eq:Zvt} using the last two equations.
\end{proof}

\begin{remark}
	Note that the discretization scheme \eqref{eq:Z-AR} has an autoregressive form. The difference between the popular autoregressive time-series models and \eqref{eq:Z-AR}, however, is that the latter has both a time varying variance of error and a time varying autoregressive coefficient. In particular, $(Z_t)_{0\le t\le T}$ can never be stationary since it is a Brownian bridge.	
\end{remark}

\begin{remark}\label{rem:Z-Gaussian}
	An argument similar to the proof of Proposition \ref{prop:Z-Sim} yields that, for $t\in[0,T]$,
	\begin{align}
		\Zv_t = \diag\left(\left(\frac{T_i-t}{T_i}\right)^{\kap_i}\right) \Zv_0
		&{}+ \diag\big(\phi_i(t)\big)\mv\\*
		&{}+ \int_0^t \diag\left(\left(\frac{T_i-t}{T_i-s}\right)^{\kap_i}\right)\,\Sigt_{\Zv}\dd \Wv_s,
	\end{align}
	in which, with a slight abuse of notation, we have defined
	\begin{align}
		\phi_i(t):=
		\begin{cases}\displaystyle
			\frac{1}{\kap_i-1}\left(T_i-t -\frac{(T_i-t)^{\kap_i}}{(T_i)^{\kap_i-1}}\right);&\quad \kap_i\ne 1,\vspace{1ex}\\\displaystyle
			(T_i-t)\ln\left(\frac{T_i}{T_i-t}\right);&\quad \kap_i=1.
		\end{cases}
	\end{align}
	In particular, conditional on $\Zv_0=\zv$, the process $(\Zv_t)_{0\le t\le T}$ is a Gaussian process with the mean function
	\begin{align}
		\Eb(\Zv_t|\Zv_0=\zv)=\diag\left(\left(\frac{T_i-t}{T_i}\right)^{\kap_i}\right) \Zv_0
		+ \diag\big(\phi_i(t)\big)\mv,
	\end{align}
	for $0\le t\le T$, and the covariance function
	\begin{align}
		\Cov\left(Z_{s,i}, Z_{t,j}\middle|\Zv_0=\zv\right) = \beta_{ij}\int_{0}^{s\wedge t}\left(\frac{T_i-s}{T_i-u}\right)^{\kap_i}\left(\frac{T_j-t}{T_j-u}\right)^{\kap_j}\dd u,
	\end{align}
	% Recall that $\beta_{ij}$ is the $ij$-th element of the matrix $\Sig_\Zv :=\Sig_\Fv+\Sig_\Sv-\Sig_{\Fv\Sv}-\Sig_{\Fv\Sv}^\top$.
	for $t,s\in[0,T]$.
\end{remark}

Next, we derive a discretization scheme for the asset prices $\Sv$ and $\Fv$. To simplify the scheme, we have assumed that $\etav_S(t)\equiv\0$. Then, an exact discretization scheme for $(\Sv_t,\Fv_t)_{0\le t\le T}$ is given by the following corollary of Proposition \eqref{prop:Z-Sim}.

\begin{corollary}\label{cor:Sim}
	Assume that $\eta_{i,S}=0$, for $i\in\{1,\dots, N\}$. Let $(\Fv_t,\Sv_t,\Zv_t)_{0\le t\le T}$ satisfy \eqref{eq:Z}--\eqref{eq:Si}. Then, for $n\in\{1,\dots,M\}$,
\begin{align}\label{eq:Sim-FSZ}
\begin{cases}
	\ln\left(\Sv_{\frac{nT}{M}}\right) = \ln\left(\Sv_{\frac{(n-1)T}{M}}\right) + \frac{T}{M} \muvt_\Sv + \ev_{n,\Sv},\vspace{1em}\\*
	\diag\left(\left(\frac{T_i}{T_i-\frac{n\,T}{M}}\right)^{\kap_i}\right)\Zv_{\frac{nT}{M}} =\\* \hspace{3em}\diag\left(\left(\frac{T_i}{T_i-\frac{(n-1)\,T}{M}}\right)^{\kap_i}\right)\Zv_{\frac{(n-1)T}{M}} + \diag(\phi_{n,i})\mv + \ev_{n,\Zv},\vspace{1em}\\*
	\log(F_{\frac{n T}{M},i}) = Z_{\frac{n T}{M},i} + \log(S_{\frac{n T}{M},i}) + r \left(T_i -\frac{n T}{M} \right),
\end{cases}
\end{align}
	in which $\phi_{n,i}$ is given by \eqref{eq:phi}, $\muvt_\Sv:=\left(\mu_{1,S}-\frac{1}{2}\sig_{1,S}^2,\dots,\mu_{N,S}-\frac{1}{2}\sig_{N,S}^2\right)^\top$, and $\{\ev_n=(\ev_{n,\Sv}, \ev_{n,\Zv})\}_{n=1}^M$ are independent Gaussian random variables such that $\Eb(\ev_n)=\0$, $\Eb(\ev_{n,\Sv}\ev_{n,\Sv}^\top)=\frac{T}{M}\Sig_\Sv$, $\Eb(\ev_{n,\Sv}\ev_{n,\Zv}^\top)=\left(\Sig_{\Fv\Sv}^\top - \Sig_{\Sv}\right)\diag(\phi_{n,i})$, and $\Eb(\ev_{n,\Zv}\ev_{n,\Zv}^\top)=[\eta_{ij;n}]_{N\times N}$ with $\eta_{ij;n}$ given by \eqref{eq:eta}.\vspace{1em}
\end{corollary}

\begin{proof}
	Under the assumption $\eta_{i,S}=0$, $i\in\{1,\dots, N\}$, \eqref{eq:Si} becomes a geometric Brownian motion (GBM),
	\begin{align}
		\dd \Sv_t = \diag(\Sv_t)\left[\muv_\Sv \dd t + [\Sigt_{\Sv\Fv}, \Sigt_\Sv] \dd \Wv_t\right].
	\end{align}
	This GBM is discretized using the first equation in \eqref{eq:Sim-FSZ}
	% \begin{align}
	% 	\ln\left(\Sv_{\frac{nT}{M}}\right) = \ln\left(\Sv_{\frac{(n-1)T}{M}}\right) + \frac{T}{M} \left(\mu_{1,S}-\frac{1}{2}\sig_{1,S}^2,\dots,\mu_{N,S}-\frac{1}{2}\sig_{N,S}^2\right)^\top + \ev_{n,\Sv},
	% \end{align}
	in which 
	\begin{align}\label{eq:eS}
		\ev_{n,\Sv} := [\Sigt_{\Sv\Fv}, \Sigt_\Sv]\left(\Wv_{\frac{n T}{M}}-\Wv_{\frac{(n-1)T}{M}}\right).
	\end{align}
	The result then follows from the discretization of $\Zv_t$ in Proposition \ref{prop:Z-Sim} and its proof. Specifically, the relationship between $\ev_{n,\Sv}$ and $\ev_{n,\Zv}$ follows from \eqref{eq:eZ} and \eqref{eq:eS}.
\end{proof}

\begin{figure}[p]
\centerline{
% \fbox{
\adjustbox{trim={0.0\width} {0.0\height} {0.0\width} {0.0\height},clip}
{\includegraphics[scale=0.23, page=1]{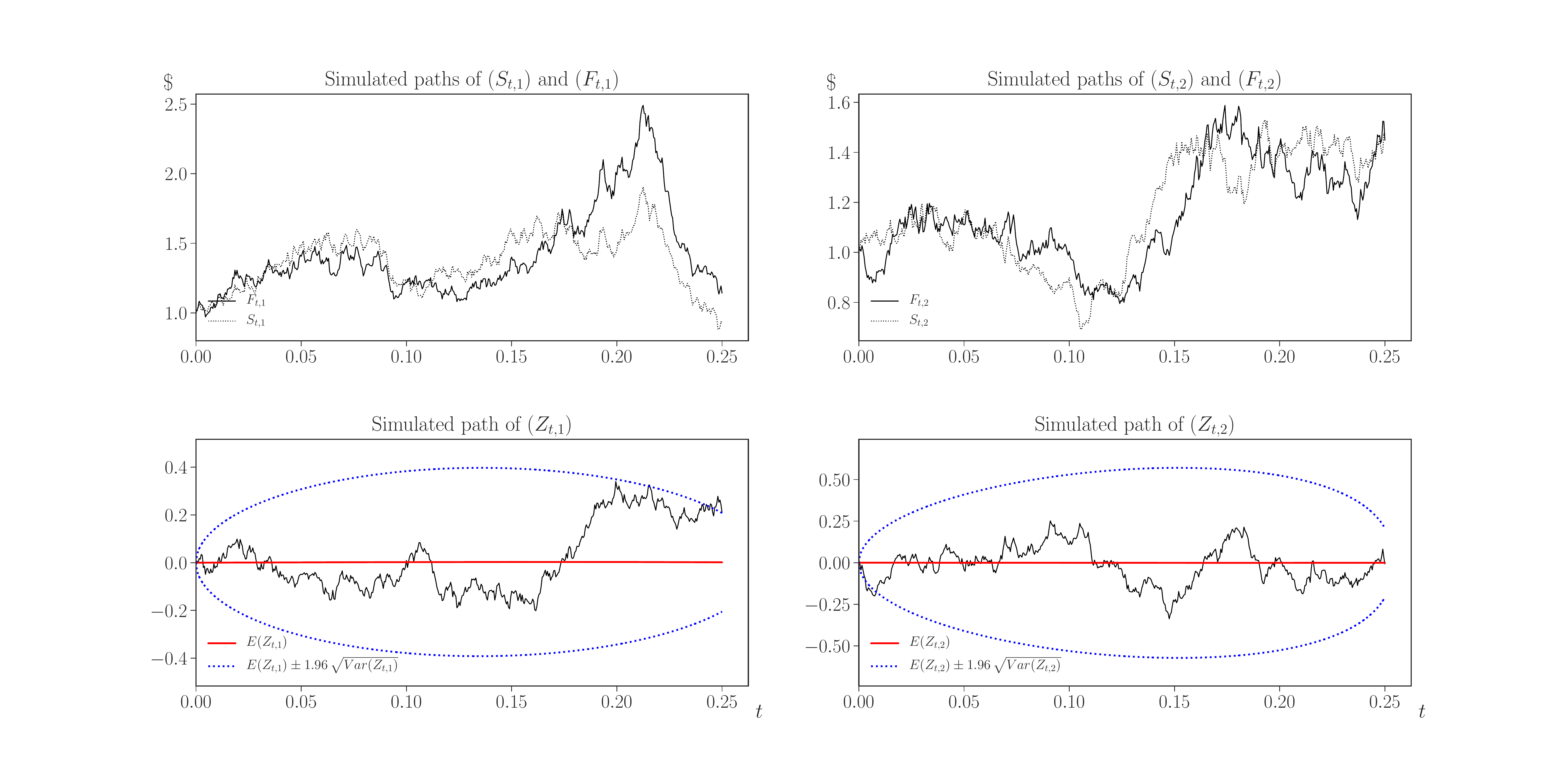}}
% }
}
\centerline{
% \fbox{
\adjustbox{trim={0.0\width} {0.0\height} {0.0\width} {0.0\height},clip}
{\includegraphics[scale=0.32, page=1]{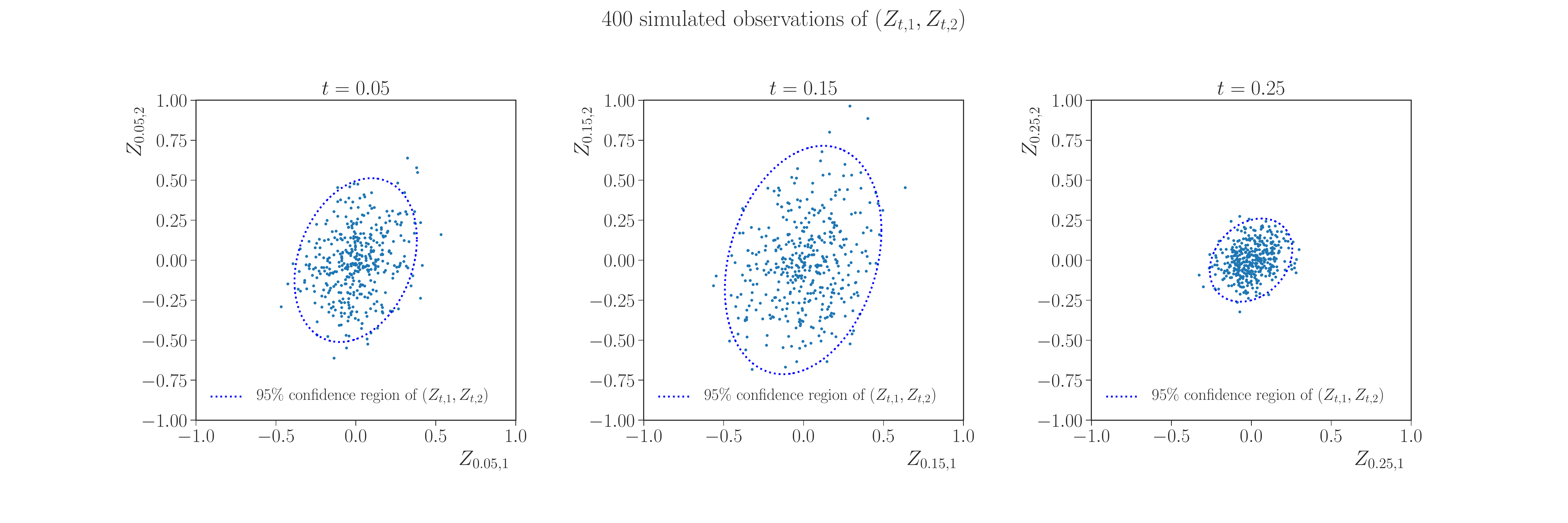}}
% }
}
\caption{\emph{Top:} Paths of $\Sv$ and $\Fv$ for two pairs of futures and spot (i.e. $N=2$) for one simulated scenario generated from \eqref{eq:Sim-FSZ}.
\emph{Middle:} The corresponding paths of $\Zv$.
\emph{Bottom:} Simulated values of $(Z_{t,1}, Z_{t,2})$ based on 400 Simulated paths observed at three different times: $t=0.05$ (left), $t=0.15$ (middle), and $t=0.25$ (right).
In the middle and bottom rows, the dotted lines represent the border of the 95\% confidence region for $(Z_{t,1}, Z_{t,2})$ based on Remark \ref{rem:Z-Gaussian} and conditional on $\Zv_0=(0,0)$.
\emph{Parameters:} $r=0.01$, $\muv_\Fv=(0.12,0.13)$, $\muv_\Sv=(0.1, 0.15)$, $(\eta_{1,F}, \eta_{2,F})=(-1,-1.5)$, $(\eta_{1,S}, \eta_{2,S})=(0,0)$, $T=0.25$, $(T_1,T_2)=(0.27, 0.26)$, and $\Sig = \protect\begin{bmatrix}
	1.0 & 0.3 & 0.7 & 0.2\protect\\
	0.3 & 1.0 & 0.1 & 0.5\protect\\
	0.7 & 0.1 & 1.0 & 0.2\protect\\
	0.2 & 0.5 & 0.2 & 1.0
\protect\end{bmatrix}$.\vspace{1em}
\label{fig:sim-FSZ}}
\vspace{1em}
\end{figure}

Figure \ref{fig:sim-FSZ} illustrates the simulated paths of $\Zv$, $\Sv$, and $\Fv$ for two pairs of futures and underlying assets (i.e. $N=2$) based on the discretization scheme \eqref{eq:Sim-FSZ}. The middle plots for $(Z_{t,1})$ and $(Z_{t,2})$ also show the 95\% confidence intervals of the log-bases. These plots showcase two characteristics of the log-bases. Firstly, they are mean-reverting in that any deviation from their mean is corrected. Secondly, they partially converge to zero at the end of the trading horizon ($T=0.25$) as evident by narrowing of the confidence intervals. Indeed, $(Z_{t,1})$ and $(Z_{t,2})$ are Brownian bridges that converge to zero at $T_1=0.27$ and $T_2=0.254$, respectively. As emphasized before, this convergence is not realized since trading stops at $T=0.25$.

The bottom plots of Figure \ref{fig:sim-FSZ} illustrate the joint behavior of $\Zv_t=(Z_{t,1},Z_{t,2})$. It shows the scatter plots of values of $(\Zv_t)$ at three different time frames (i.e. $t=0.05$, $t=0.15$, and $t=T=0.25$) for 400 simulated paths generated by the discretization scheme \eqref{eq:Sim-FSZ}. Since, $\Zv_t|\Zv_0$ is a multivariate normal random variable, we have that
\begin{align*}
    \big(\Zv_t|\Zv_0 - \Eb(\Zv_t|\Zv_0)\big)^\top \Cov(\Zv_t|\Zv_0)^{-1}\big(\Zv_t|\Zv_0 - \Eb(\Zv_t|\Zv_0)\big)
\end{align*}
has chi-squared distribution with $N$ degrees of freedom. This relationship is utilized for obtaining the 95\% confidence regions of $\Zv_t$ in Figure \ref{fig:sim-FSZ} represented by dashed blue ellipses. These plots also illustrate partial convergence of log-bases at the end of time horizon.

\section{Optimal trading of a basket of futures}\label{sec:Futures}

Most speculative futures trading strategies involve only the futures and not the underlying assets. Indeed, many underlying assets cannot be easily traded speculatively (e.g. agricultural commodities). In other cases where the underlying is tradable (e.g. precious metals or stock futures), futures are traded \emph{``as a proxy for''} their underlying assets because of their attractive trading characteristics such as leverage and the ease of taking short positions.

Motivated by these practical applications, we formulate and solve the optimal investment problem faced by a trader who invest in the market but only trades the futures contracts. This is an optimal investment problem in an incomplete market, since the underlying assets are not tradable.

Let $\tht_{t,i}$ be the notional value\footnote{That is, the number of futures contracts held multiplied by the futures price.} invested in $F_i$ at $t\in[0,T]$. The trader's wealth, denoted by $(X_t)_{0\le t\le T}$, then satisfies
\begin{align}\label{eq:BudgetF}
d X_t &= \sum_{i=1}^{2} \tht_{t,i} \frac{\dd F_{t,i}}{F_{t,i}} + r X_t \dd t\\
&= \left[r X_t + \thtv_t^\top \left(\muv_\Fv + \etav_F(t)\Zv_t\right)\right]\dd t
  + \thtv_t^\top \Sigt_{\Fv} \dd \Wv_{t,1},
\end{align}
with $X_0=x>0$. An $(\Fc_t)_{0\le t\le T}$-adapted process $\big(\thtv_t^\top=(\tht_{t,1},\ldots,\tht_{t,N})\big)_{0\le t\le T}$ is an \emph{admissible trading strategy} if it satisfies the following conditions,
\begin{enumerate}
	\item[(i)] $\sum_{i=1}^N\int_0^T\left[\tht_{t,i}^2 + |Z_{t,i}\tht_{t,i}|\right] \dd t <\infty$, $\Pb$-a.s.; and,
	\item[(ii)] $X_t>0$ $\Pb-a.s.$ for all $t\in[0,T]$, where $(X_t)_{0\le t\le T}$ is given by \eqref{eq:BudgetF}.
\end{enumerate}
The set of all admissible policy is denoted by $\Ac_F$.

We assume that the trader's objective is to maximize her expected utility of terminal wealth. Thus, the trader faces the following stochastic control problem
\begin{align}\label{eq:VF-F}
	\VF(t,x,\zv) := \sup_{\thtv\in\Ac_F} \Eb_{t,x,\zv}\left(\frac{X_T^{1-\gam}}{1-\gam}\right);\quad (t,x,\zv)\in[0,T]\times\Rb_+\times\Rb^N.
\end{align}
in which $\Eb_{t,x,\zv}(\cdot) := \Eb(\cdot|X_t=x, \Zv_t=\zv)$ and $\gam>1$ is the trader's constant relative risk aversion parameter.

As the reader may have noticed, we only consider power utilities with $\gam>1$, that is, we have excluded power utilities that are more risk seeking than the logarithmic utility. The analysis of the case $\gam\in(0,1)$ is more intricate and, in particular, involves identifying the so-called \emph{nirvana solutions} where the expected utility becomes infinite. We briefly discuss nirvana solutions in Figure \ref{fig:gam_sens} below. For a more complete treatment of the subject, we refer the reader to \cite{AngoshtariLeung2019} for a detailed analysis of nirvana solutions in a special case of our market model where there is only one pair of futures and underlying assets. To keep the discussion less technical, however, we have decided to focus on the well-posed case $\gam>1$.

%  are tradable by imposing the constraint
% \begin{align}
% 	\pi_{t,i}=0;\quad 0\le t\le T, i\in\{1,\dots,N\},
% \end{align}
% in \eqref{eq:Budget}. Let $\Ac_F$ denote the set of all admissible policies that only trade the futures, that is, the set of all $(\Fc_t)_{0\le t\le T}$-adapted process $\thtv_t^\top=(\tht_{t,1},\ldots,\tht_{t,N})$ such that,
% \begin{enumerate}
% 	\item[(i)] $\sum_{i=1}^N\int_0^T\left[\tht_{t,i}^2 + \frac{|Z_{t,i}\tht_{t,i}|)}{T_i-t}\right] \dd t <\infty$, $\Pb$-a.s.; and,
% 	\item[(ii)] $X_t>0$ $\Pb-a.s.$ for all $t\in[0,T]$, in which $(X_t)_{0\le t\le T}$ is given by
% 	% \begin{align}\label{eq:BudgetF}
% % 		d X_t &= \sum_{i=1}^{2} \tht_{t,i} \frac{\dd F_{t,i}}{F_{t,i}} + r X_t \dd t\\*
% % 		&= \left[r X_t + \thtv_t^\top \left(\muv_\Fv + \etav_F(t)\Zv_t\right)\right]\dd t
% % 		  + \thtv_t^\top \Sigt_{\Fv} \dd \Wv_{t,1},
% % 	\end{align}
% % 	with $X_0=x>0$.
% \end{enumerate}
% The trader's value function is then given by
% % \begin{align}\label{eq:VF-F}
% % 	\VF(t,x,\zv) := \sup_{\thtv\in\Ac_F} \Eb_{t,x,\zv}\left(\frac{X_T^{1-\gam}}{1-\gam}\right);\quad (t,x,\zv)\in[0,T]\times\Rb_+\times\Rb^N.
% % \end{align}

The following result provides the solution of this stochastic control problem.  To facilitate presentation, the following notations are used.
\begin{align}\label{eq:A}
	\AF &:= \Sig_\Fv + \Sig_{\Fv\Sv}^\top\Sig_\Fv^{-1}\Sig_{\Fv\Sv} - \Sig_{\Fv\Sv} - \Sig_{\Fv\Sv}^\top,\quad \BF := I_N - \Sig_\Fv^{-1}\Sig_{\Fv\Sv},
\end{align}
\begin{align}\label{eq:EClam}
	C :=
	\begin{pmatrix}
		I_N\\
		-I_N
	\end{pmatrix},
	\quad\text{and}\quad
	\etav(t) &:=
	\begin{pmatrix}
		\etav_F(t)\\
		\etav_S(t)
	\end{pmatrix},
\end{align}
in which $I_N$ is the $N\times N$ identity matrix and $\Sig_\Fv$ and $\Sig_{\Fv\Sv}$ are given by \eqref{eq:Cov}.\pagebreak

\begin{theorem}\label{thm:Futures}
	The following statements hold.
   	\begin{enumerate}
   		\item[(i)] The matrix Riccati differential equation below has a unique solution that is positive definite for all $\tau\in(0,T]$,
   	\begin{align}\label{eq:H-ODE-F}
   	\begin{cases}
		H'(\tau) 
		+ H(\tau)\, \big(\gam\,C^\top\Sig\,C + (1-\gam)A\big)\, H(\tau)\\*
		\hspace{3em}- 2\left(\etav(T-\tau)^\top C + \left(\frac{1}{\gam}-1\right)\etav_F(T-\tau)^\top B\right)\, H(\tau)\\*
		\hspace{3em}- \frac{\gam-1}{\gam^2} \etav_F(T-\tau)^\top \Sig_\Fv^{-1}\etav_F(T-\tau)=0;
   		\quad0\le \tau\le T,\\
   		H(0)=0_{N\times N}.
   	\end{cases}
   	\end{align}
		
   		\item[(ii)] The value function in \eqref{eq:VF-F} is given by
   		\begin{align}\label{eq:VF-sol-F}
   			\VF(t,x,\zv) &=\frac{x^{1-\gam}}{1-\gam}
   			\ee^{\gam \left(f(T-t) + \zv^\top \gv(T-t) - \frac{1}{2} \zv^\top H(T-t)\zv\right)},
   		\end{align}
   		for $(t,x,\zv)\in[0,T]\times\Rb_+\times\Rb^N$, in which $\gv(\tau)=\big(g_1(\tau),\ldots, g_N(\tau)\big)^\top$ satisfies
   		\begin{align}\label{eq:g-ODE-F}
   		\begin{cases}
   			\gv'(\tau) = \bigg[\etav(T-\tau)^\top C - \gam\,H(\tau)\,C^\top\Sig\,C\\*
			\qquad\qquad{}- (1-\gam)H(\tau)A 
		+ \left(\frac{1}{\gam}-1\right)\etav_F(T-\tau)^\top B\bigg]\,\gv(\tau)\\*
   			\qquad\qquad-H(T-\tau)\left(\mv + \left(\frac{1}{\gam}-1\right) B^\top\muv_\Fv\right)\\*
			\qquad\qquad{}+\frac{1-\gam}{\gam^2}\,\etav_F(T-\tau)^\top\Sig_\Fv^{-1}\muv_\Fv;
   			\qquad\qquad 0\le \tau\le T,\\
   			\gv(0)= 0_{N\times 1},
   		\end{cases}
   		\end{align}
   		and $f(\tau)$ is given by
   		\begin{align}\label{eq:f-integral-F}
   			f(\tau) ={}&\frac{1-\gam}{\gam}\left(r + \frac{\muv_\Fv^\top\Sig_\Fv^{-1}\muv_\Fv}{2\gam}\right)\tau\\*
   			&+\int_0^\tau \bigg[
   			\frac{1}{2}\gv(u)^\top \left((1-\gam) A + \gam\, C^\top \Sig\,C\right) \gv(u)\\
   			&\hspace{3em} + \left(\mv+\left(\frac{1}{\gam}-1\right)B^\top \muv_\Fv\right)^\top\gv(u)\\*
			&\hspace{3em}{}-\frac{1}{2}\tr\big(C^\top\Sig\, C\, H(u)\big)
   			\bigg]\dd u,
   		\end{align}
   		for $0\le \tau\le T$.
		
   		\item[(iii)] The optimal trading strategy for the futures contracts is $\Big(\thtv^*(t,X_t,\Zv_t)\Big)_{0\le t\le T}$ where
   		\begin{align}\label{eq:OptimalStrat-F}
   			\thtv^*(t,x,\zv)
   			:=
   			x
   			\bigg[
   			&\frac{1}{\gam}\Sig_{\Fv}^{-1}\muv_\Fv
   			+ B\,\gv(T-\tau) \\*
			&{} + \left(\frac{1}{\gam}\Sig_\Fv^{-1}\etav_F(t) - B\, H(T-\tau)\right)\zv
   			\bigg],
   		\end{align}
   		for $t,x,\zv\in[0,T]\times\Rb_+\times\Rb^N$.
   	\end{enumerate}\vspace{1ex}
\end{theorem}

\begin{proof}
	See Appendix \ref{app:Futures}.
\end{proof}

\begin{figure}[t]
\centerline{
%\fbox{
\adjustbox{trim={0.08\width} {0.0\height} {0.08\width} {0.0\height},clip}
{\includegraphics[scale=0.24, page=1]{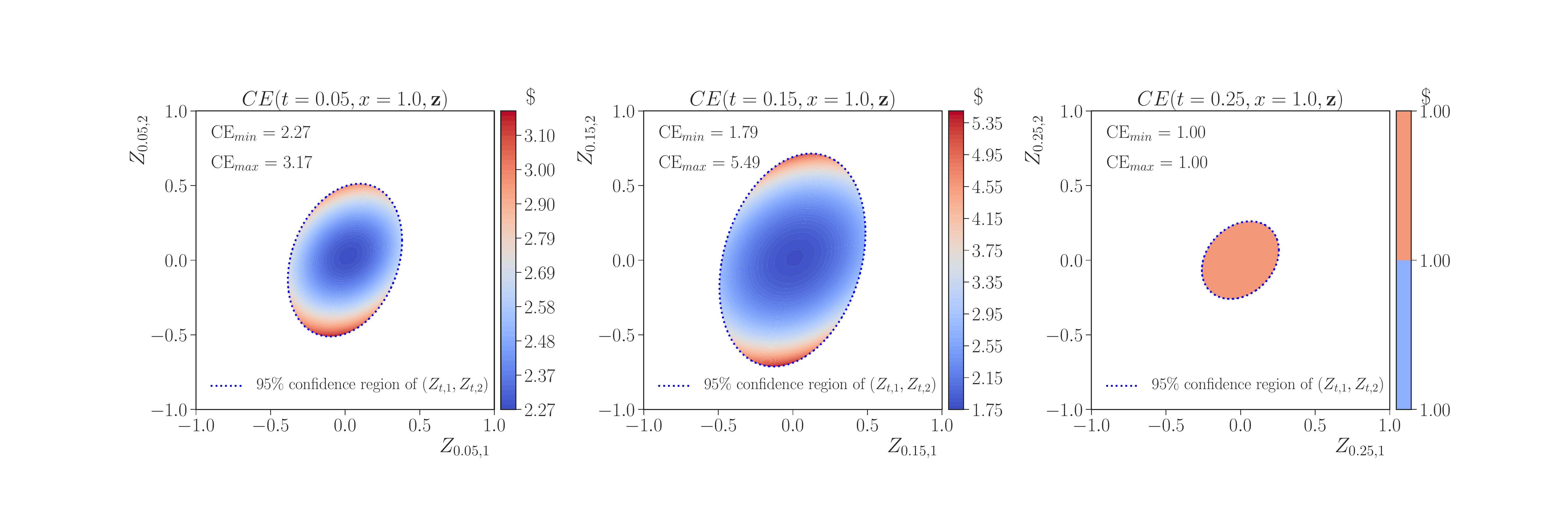}}
%}
}
\centerline{
%\fbox{
\adjustbox{trim={0.08\width} {0.0\height} {0.08\width} {0.0\height},clip}
{\includegraphics[scale=0.24, page=1]{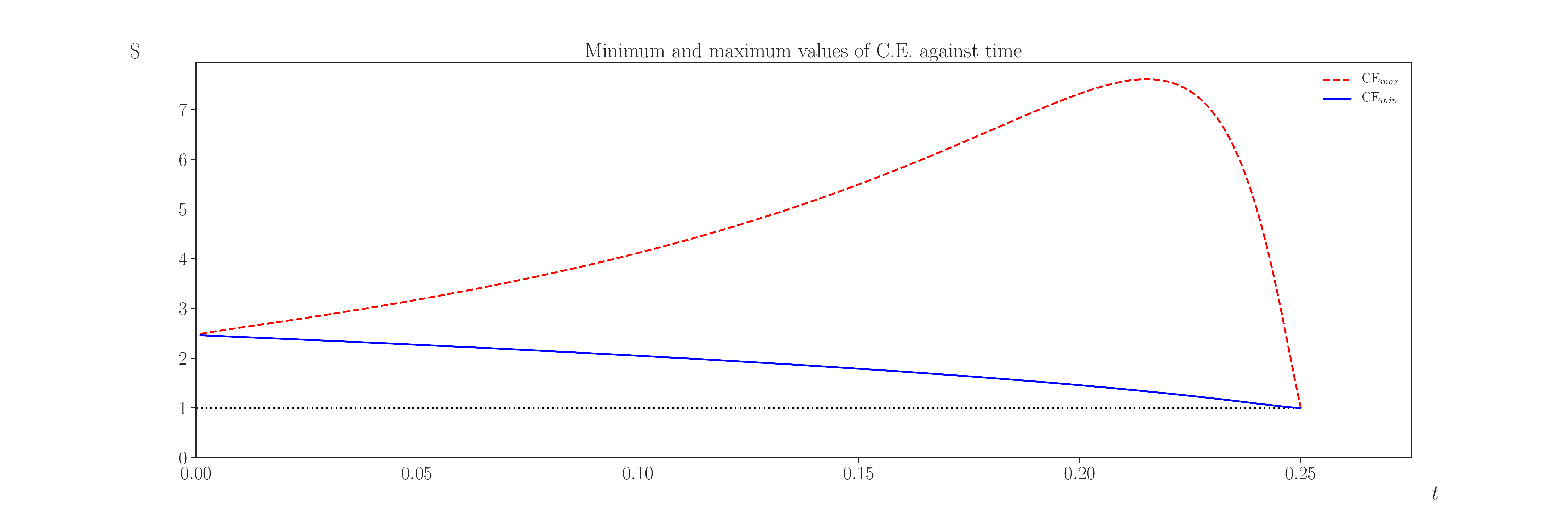}}
%}
}
\caption{\emph{Top:} Certainty equivalent (CE) values corresponding to the value function in Theorem \ref{thm:Futures} (trading only futures) at three different time instances: $t=0.05$ (left), $t=0.15$ (middle), and $t=0.25$ (right). In each plot, the value of wealth is $x=\$1$ and CE values are plotted for different values of log-bases $(Z_{t,1},Z_{t,2})$ in the 95\% confidence region as shown in Figure \ref{fig:sim-FSZ}.
\emph{Bottom:} The minimum (solid blue line) and the maximum (dashed red line) values of CE values (over the 95\% confidence region of $(Z_{t,1}, Z_{t,2})$) as a function of time.
\emph{Parameters:} $\gam=2.0$. The remaining parameters are as in Figure \ref{fig:sim-FSZ}.
\vspace{1em}
\label{fig:optimF_CE}}
\end{figure}

Obtaining the value function and the optimal policy in Theorem \ref{thm:Futures} is numerically feasible. Indeed, the matrix Riccati equation \eqref{eq:H-ODE-F} is a first order non-linear system of ordinary differential equations and can readily be solved numerically. Once \eqref{eq:H-ODE-F} is solved, \eqref{eq:g-ODE-F} becomes a first order linear system of ordinary differential equations and \eqref{eq:f-integral-F} is an integral, both of which can be easily computed. One can then find the value function and the optimal policy by \eqref{eq:VF-sol-F} and \eqref{eq:OptimalStrat-F}, respectively.

The value function is better interpreted by its corresponding certainty equivalent (CE). We denote by $\CE(t,x,\zv)$ the certainty equivalent value of the trader at the state $(t,x,\zv)$. It is the minimum amount of deterministic terminal wealth which the investor prefers over optimally investing in the market starting at the state $(t,x,\zv)$. The CE function thus satisfies
\begin{align}
	\frac{\big(\CE(t,x,\zv)\big)^{1-\gam}}{1-\gam} = V(t,x,\zv),
\end{align}
which, in turn, yields,
\begin{align}\label{eq:CE}
	\CE(t,x,\zv) = x\, \ee^{\left(\frac{\gam}{1-\gam}\right)
	\left(f(T-t) + \zv^\top \gv(T-t) - \frac{1}{2} \zv^\top H(T-t)\zv\right)}.
\end{align}

The top plots in Figure \ref{fig:optimF_CE} illustrate CE values corresponding to the value function given by Theorem \ref{thm:Futures} at three time instances. In each plot, we assume that $x=\$1$ and show $\CE(t,1,\zv)$ for different values of $\zv$ in its 95\% confidence region. Note that since $f$, $\gv$, and $H$ vanish at $T$, from \eqref{eq:CE} it follows that $\CE(t,x,\zv)=x$. This is confirmed numerically by the rightmost plot in Figure \ref{fig:optimF_CE} showing that $\CE(t,1.0,\zv)=1.0$. As expected, the figure shows that profitability is higher for the more extreme values of $\zv$ (i.e. when the log-basis deviates more from its equilibrium state). It also shows that potential for profitability increases and then decreases as time pass by.

The latter point is highlighted in the bottom plot of Figure \ref{fig:optimF_CE} which shows the minimum and maximum values of $\CE(t,1.0,\zv)$ over the 95\% confidence region of $\zv$ and as a function of $t$. As it can be seen, the minimum CE value decreases with time, while the maximum CE value (over the 95\% confidence region) first increases to a maximum point and then decreases to 1 (the value of $x$). This behavior can be explained as follows. As we approach the end of the trading horizon, there is less time to make profit. The force of mean reversion, however, is stronger near the end of the trading horizon (recall that, from \eqref{eq:Bb} the mean reversion coefficient is $\kap_i/(T_i-t)$). At the beginning, the increase in the mean reversion rate plays a more prominent rule and causes the maximum CE value to increase. At the end, the reduction of trading time has a more significant affect and causes the maximum CE value to decrease. Thus, the combined effect of these two opposing forces causes the increase and then decrease in the maximum value of certainty equivalent.

\begin{figure}[t]
\centerline{
% \fbox{
\adjustbox{trim={0.07\width} {0.09\height} {0.06\width} {0.1\height},clip}
{\includegraphics[scale=0.3, page=1]{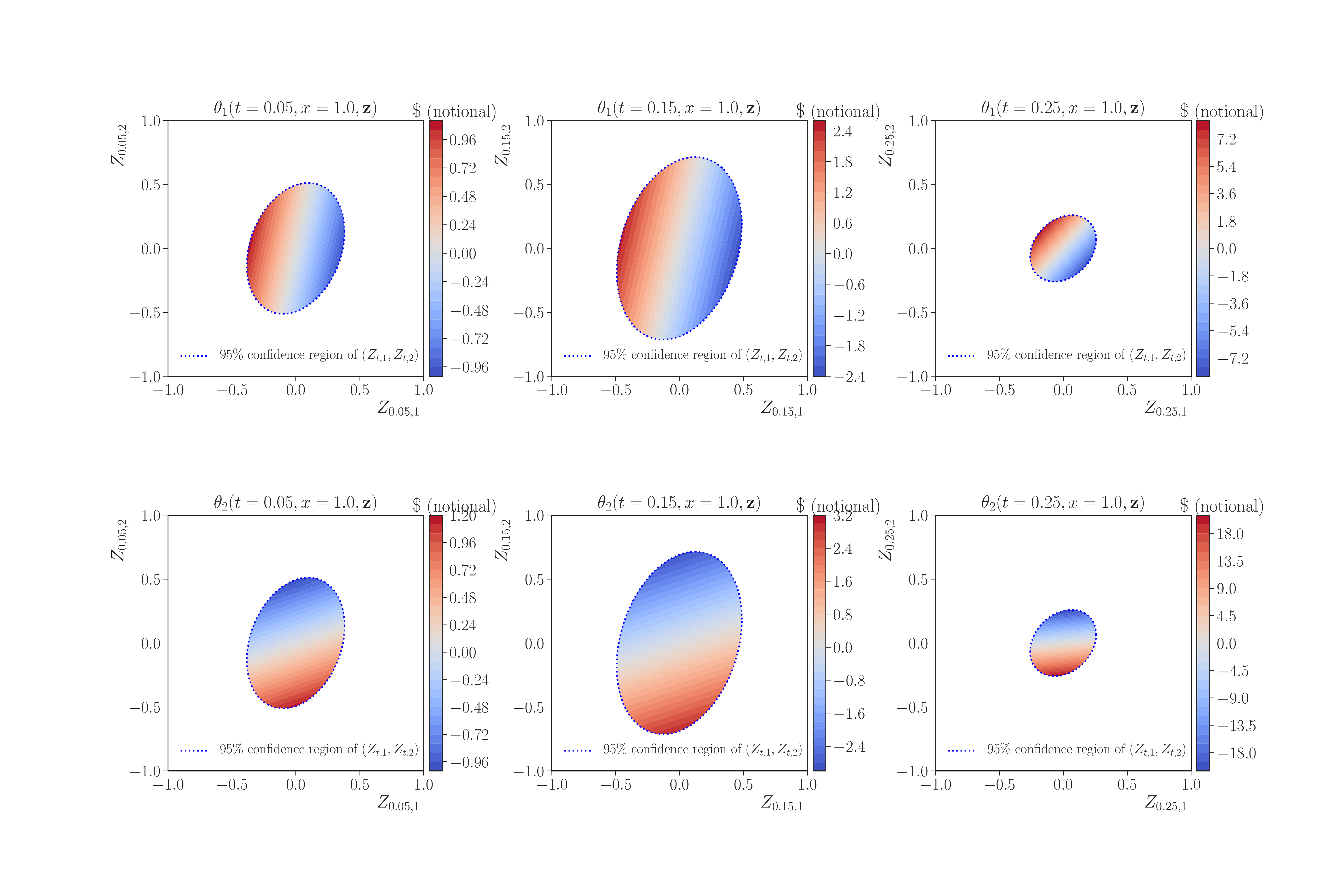}}
% }
}
\caption{Optimal positions for the futures only strategy, at three time instances: $t=0.05$ (left column), $t=0.15$ (middle column), and $t=0.25$ (right column). In each plot, the optimal positions are plotted for fixed value of wealth $x=1$ and different values of the log-bases $(Z_{t,1}, Z_{t,2})$ in their 95\% confidence region.
\vspace{1em}
\label{fig:optimF_theta}}
\end{figure}

Figure \ref{fig:optimF_theta} shows the optimal positions in the futures for three time instances: $t=0.05$ (left column), $t=0.15$ (middle column), and $t=0.25$ (right column). We can make a few observations based on the plots. Firstly, note that the optimal trading strategy for the futures contracts is a general relative value strategy driven by the log basis. In particular, for sufficiently large positive (resp. negative) values of $Z_{t,i}$, it is optimal to take a short (resp. long) position in the $i$-th futures contract. As mentioned earlier, large positive (resp. negative) values of $Z_{t,i}$ indicate that the market is contangoed (resp. backwardated). Therefore, our optimal strategy is consistent with the theory of storage and hedging pressure hypostasis, as discussed in Section \ref{sec:Intro}.

Secondly, the value of the log-basis of a pair of futures and underlying has an effect on the optimal positions of the other pair. This is evident in the plots since the level curves are not vertical or horizontal. For example, the plot for $\tht_1$ (the top-left plot) indicates that the optimal position in $F_1$ changes value if we change $Z_2$ while keep $Z_1$ fixed. This dependence signifies the importance of modeling the futures jointly rather than individually. 

Finally, note that the extreme positions (i.e. the larger positions near the edge of the oval confidence region) become larger as we approach the end of the trading horizon. This is expected since, as mentioned before, the force of mean reversion is stronger near the end of the trading horizon. Therefore, deviations from equilibrium are corrected faster and it is optimal to take larger positions in case of such deviation near the end of the trading horizon.

\section{Optimal trading strategy when the underlying assets are tradable}\label{sec:SpotFutures}
In this section, we consider the scenario in which the trader invests both in the futures contracts and the underlying assets. This problem has limited practical relevance, since most speculative traders don't maintain positions on both futures and the spot.\footnote{A notable exception is ``\emph{basis trading}'', see \cite{AngoshtariLeung2019} for further discussion.} Our main objective for considering this problem is to quantify how much value one would gain by trading the spot, while acknowledging that doing so is not always feasible (or practical). Note, also, that this scenario results in an optimal investment problem in a complete market since the underlying assets are now tradable. 

Let $\Ac$ denote the set of all admissible policies that trade the futures contracts and the underlying assets. In other words, it is the set of all $(\Fc_t)_{0\le t\le T}$-adapted process $\big(\Thtv^\top_t:=(\tht_{t,1},\ldots, \tht_{t,N}, \pi_{t,1},\dots,\pi_{t,N})\big)_{0\le t\le T}$ such that, 
\begin{enumerate}
	\item[(i)] $\sum_{i=1}^N\int_0^T\left[\tht_{t,i}^2 + \pi_{t,i}^2 + |Z_{t,i}|(|\pi_{t,i}|+|\tht_{t,i}|)\right] \dd t <\infty$, $\Pb$-a.s.; and,
	\item[(ii)] $X_t>0$ $\Pb-a.s.$ for all $t\in[0,T]$, where $(X_t)_{0\le t\le T}$ is given by
	\begin{align}\label{eq:Budget}
		d X_t &= \sum_{i=1}^{N} \tht_{t,i} \frac{\dd F_{t,i}}{F_{t,i}} + \sum_{i=1}^N \pi_{t,i}\frac{\dd S_{t,i}}{S_{t,i}} + r\left(X_t - \sum_{i=1}^N \pi_{t,i}\right) \dd t\\*
		% &= \left[r X_t + \thtv_t^\top \left(\alv + \etav_1(t)\Zv_t\right) + \piv_t^\top \left(\muv-\rv+\etav_2(t)\Zv_t\right)\right]\dd t\\
		% &\qquad{}+ (\thtv_t^\top\Sigt_\Fv + \piv_t^\top \Sigt_{\Sv\Fv})\dd \Wv_{t,1} + \piv_t^\top \Sigt_\Sv \dd \Wv_{t,2}\\*
		&= \left[r X_t + \Thtv_t^\top \left(
		\begin{pmatrix}
			\muv_\Fv\\
			\muv_\Sv-\rv
		\end{pmatrix} +
		\begin{pmatrix}
		\etav_F(t)\\
		\etav_S(t)
		\end{pmatrix}\Zv_t\right)\right]\dd t + \Thtv_t^\top \Sigt \dd \Wv_t,
	\end{align}
	with $X_0=x>0$.
\end{enumerate}
Here, $\pi_{t,i}$ (resp. $\tht_{t,i}$) is the cash amount (resp. notional value) invested in $S_i$ (resp. $F_i$) at $t\in[0,T]$ and we have defined $\rv^\top := (r,\ldots,r)_{1\times N}$.

The trader's value function is then given by
\begin{align}
	V(t,x,\zv) := \sup_{\Thtv\in\Ac} \Eb_{t,x,\zv}\left(\frac{X_T^{1-\gam}}{1-\gam}\right);\quad (t,x,\zv)\in[0,T]\times\Rb_+\times\Rb^N,
\end{align}
where, as before, it is assumed that $\gam>1$. The following theorem characterizes the value function and the optimal trading strategy. In its statement, we have used the notation:
\begin{align}
	\muv :=
	\begin{pmatrix}
		\muv_\Fv\\
		\muv_\Sv-\rv
	\end{pmatrix}.
\end{align}
Recall, also, that $\etav$ and $C$ are given by \eqref{eq:EClam}.

\begin{figure}[p]
\centerline{
% \fbox{
\adjustbox{trim={0.0\width} {0.0\height} {0.0\width} {0.0\height},clip}
{\includegraphics[scale=0.3, page=1]{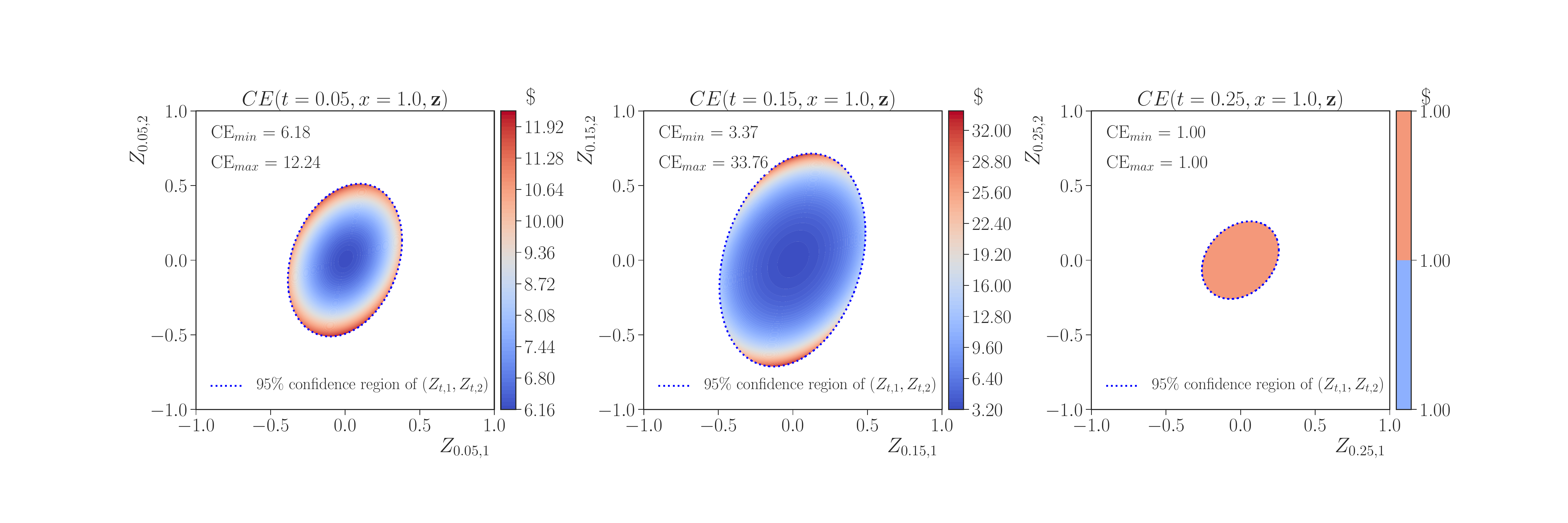}}
% }
}
\centerline{
% \fbox{
\adjustbox{trim={0.0\width} {0.0\height} {0.0\width} {0.0\height},clip}
{\includegraphics[scale=0.3, page=1]{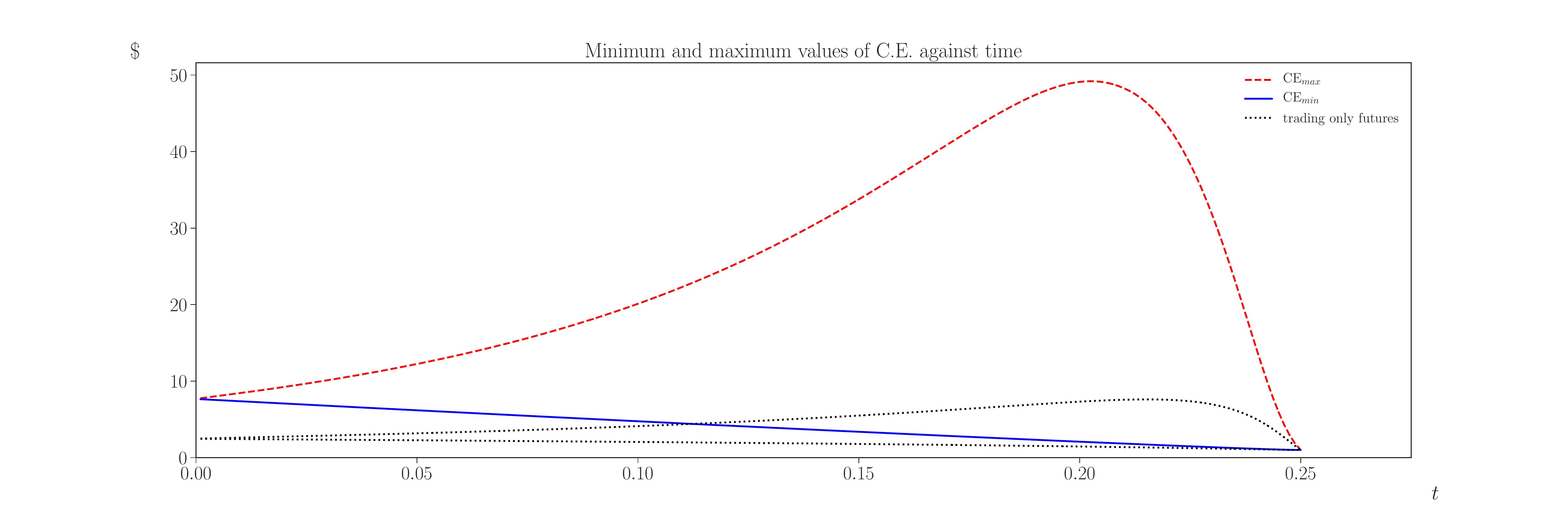}}
% }
}
\caption{\emph{Top:} Certainty equivalent (CE) values corresponding to the value function in Theorem \ref{thm:VF} (trading futures and spots) at three different time instances: $t=0.05$ (left), $t=0.15$ (middle), and $t=0.25$ (right). In each plot, the value of wealth is $x=\$1$ and CE values are plotted for different values of log-bases $(Z_{t,1},Z_{t,2})$ in the 95\% confidence region as shown in Figure \ref{fig:sim-FSZ}.
\emph{Bottom:} The minimum (solid blue line) and the maximum (dashed red line) values of CE values when trading both futures and spots (over the 95\% confidence region of $(Z_{t,1}, Z_{t,2})$) as a function of time. The dotted black lines are the corresponding values taken from the bottom plot of Figure \ref{fig:optimF_CE}, i.e. the minimum and maximum CE values when trading only the futures. The plot shows considerable gain in CE values when the underlying is tradable.
\emph{Parameters:} $\gam=2.0$. The remaining parameters are as in Figure \ref{fig:sim-FSZ}.
\vspace{1em}
\label{fig:optimFS_CE}}
\vspace{1em}
\end{figure}

\begin{theorem}\label{thm:VF}
	The following statements are true.
	\begin{enumerate}
		\item[(i)] The matrix Riccati differential equation
	\begin{align}\label{eq:H-ODE}
	\begin{cases}
		H'(\tau) +
		H(\tau) C^\top \Sig\, C\, H(\tau)
		- \frac{2}{\gam}\,\etav(T-\tau)^\top C\, H(\tau)\\
		\hspace{3em}- \frac{\gam-1}{\gam^2}\,\etav(T-\tau)^\top \Sig^{-1}\etav(T-\tau)=0;
		\quad0\le \tau\le T,\\
		H(0)=0_{N\times N},
	\end{cases}
	\end{align}
		has a unique symmetric solution $H(\tau)$ that is positive definite for all $\tau\in(0,T]$.
		
		\item[(ii)] The value function is given by
		\begin{align}\label{eq:VF-sol}
			V(t,x,\zv) &=\frac{x^{1-\gam}}{1-\gam}
   			\ee^{\gam \left(f(T-t) + \zv^\top \gv(T-t) - \frac{1}{2} \zv^\top H(T-t)\zv\right)},
		\end{align}
		for $(t,x,\zv)\in[0,T]\times\Rb_+\times\Rb^N$, in which $\gv(\tau)=\big(g_1(\tau),\ldots, g_N(\tau)\big)^\top$ satisfies
		\begin{align}\label{eq:g-ODE}
		\begin{cases}
			\gv'(\tau) = \left[\frac{1}{\gam} \etav(T-\tau)^\top - H(\tau)C^\top\Sig\right]C\,\gv(\tau)\\*
			\qquad\qquad-H(T-\tau)\left(\mv + \left(\frac{1}{\gam}-1\right) C^\top\muv\right)\\*
			\qquad\qquad{}+\frac{1-\gam}{\gam^2}\,\etav(T-\tau)^\top\Sig^{-1}\muv;
			\qquad\qquad0\le \tau\le T,\\
			\gv(0)= 0_{N\times 1},
		\end{cases}
		\end{align}
		and $f(\tau)$ is given by
		\begin{align}\label{eq:f-integral}
			f(\tau) ={}&\frac{1-\gam}{\gam}\left(r + \frac{\muv^\top\Sig^{-1}\muv}{2\gam}\right)\tau\\*
			&+\int_0^\tau \bigg[
			\frac{1}{2}\gv(u)^\top C^\top \Sig\, C\, \gv(u)
			-\frac{1}{2}\tr\big(C^\top\Sig\, C\, H(u)\big)\\*
			&\hspace{5em}{}
			+ \left(\mv^\top+\left(\frac{1}{\gam}-1\right)\muv^\top C\right)\,\gv(u)
			\bigg]\dd u,
		\end{align}
		for $0\le \tau\le T$.
		
		\item[(iii)] The optimal policy is $\Big(\Thtv^*(t,X_t,\Zv_t)\Big)_{0\le t\le T}$ where
		\begin{align}\label{eq:OptimalStrat}
			\Thtv^*(t,x,\zv)
			:=
			x
			\bigg[
			&\frac{1}{\gam}\Sig^{-1}\muv
			+ C\, \gv(T-t)\\* 
			&{}+ \left(\frac{1}{\gam}\Sig^{-1} \etav(t) - C\, H(T-t)\right)\zv
			\bigg],
		\end{align}
		for $t,x,\zv\in[0,T]\times\Rb_+\times\Rb^N$.
	\end{enumerate} \vspace{1ex}
\end{theorem}
\begin{proof}
	See Appendix \ref{app:VF}.\vspace{1ex}
\end{proof}

Figure \ref{fig:optimFS_CE} is the counterpart of Figure \ref{fig:optimF_CE} and provide information about the CE values when both the futures and the underlying assets are traded. The plots show the same patterns that we saw when trading only futures, namely, higher CE values for extreme values of $\Zv$, as well as the increase and then decrease pattern in the maximum CE value as we approach the end of the trading horizon. It is, however, evident that one loses significant value when the underlying assets are not traded. This is evident in the bottom plot of Figure \ref{fig:optimF_CE}, which includes the CE range when trading only the futures (the dotted black lines) from the bottom plot of Figure \ref{fig:optimF_CE}. For example, by comparing the top middle plots in Figures \ref{fig:optimF_CE} and \ref{fig:optimFS_CE}, we found that at $t=0.15$, the maximum CE value drops from \$33.76 when the underlying assets are traded, to \$5.49 when they are not. Similarly, the minimum CE value drop from \$3.37 to \$1.79.

Figure \ref{fig:optimFS_theta} shows the optimal positions in the futures (the top two rows of plots) and the underlying assets (the bottom two rows of plots) for three time instances: $t=0.05$ (left column), $t=0.15$ (middle column), and $t=0.25$ (right column). As in the case of trading only futures (i.e. compare with Figure \ref{fig:optimF_theta}), the value of the log-basis of a pair of futures and underlying has an effect on the optimal positions of the other pair. Furthermore, the extreme positions (i.e. the larger positions near the edge of the oval confidence region) become larger as we approach the end of the trading horizon. Finally, note that the optimal positions for each pair of futures and underlying are generally short-long positions that are common for convergence trading strategies. For example, whenever we short $F_1$ (i.e. $\tht_1<0$), we long $S_1$ (i.e. $\pi_1>0$) and vice versa. The situation is the same for $F_2$ and $S_2$.

\begin{figure}[p]
\centerline{
% \fbox{
\adjustbox{trim={0.07\width} {0.1\height} {0.06\width} {0.1\height},clip}
{\includegraphics[scale=0.3, page=1]{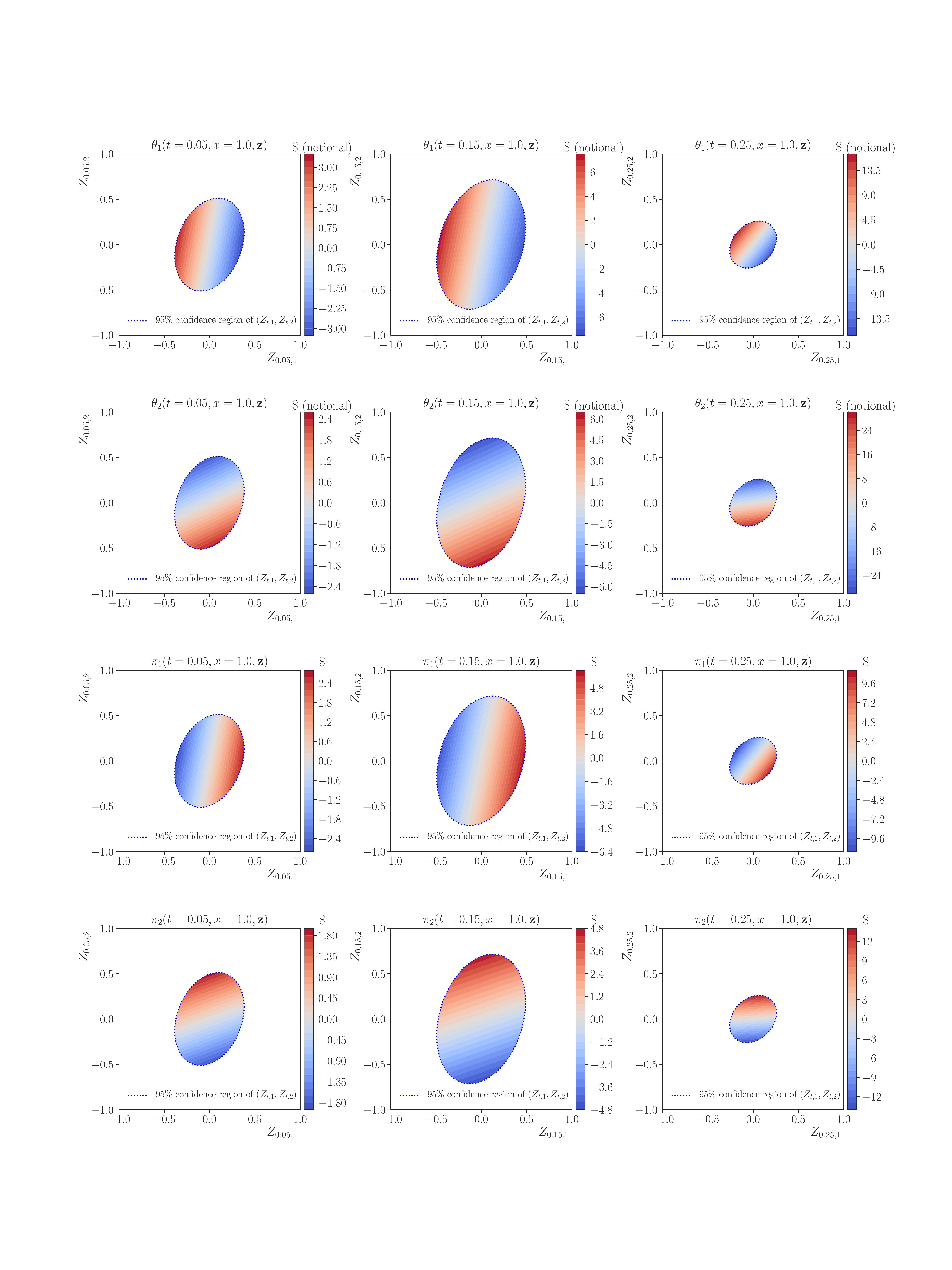}}
% }
}
\caption{Optimal positions in the futures (the top two rows) and the underlying assets (the bottom two rows) for three time instances: $t=0.05$ (left column), $t=0.15$ (middle column), and $t=0.25$ (right column). In each plot, the optimal positions are plotted for fixed value of wealth $x=1$ and different values of the log-bases $(Z_{t,1}, Z_{t,2})$ in their 95\% confidence region.
\vspace{1em}
\label{fig:optimFS_theta}}
\end{figure}

\begin{figure}[t]
\centerline{
% \fbox{
\adjustbox{trim={0.0\width} {0.0\height} {0.0\width} {0.0\height},clip}
{\includegraphics[scale=0.3, page=1]{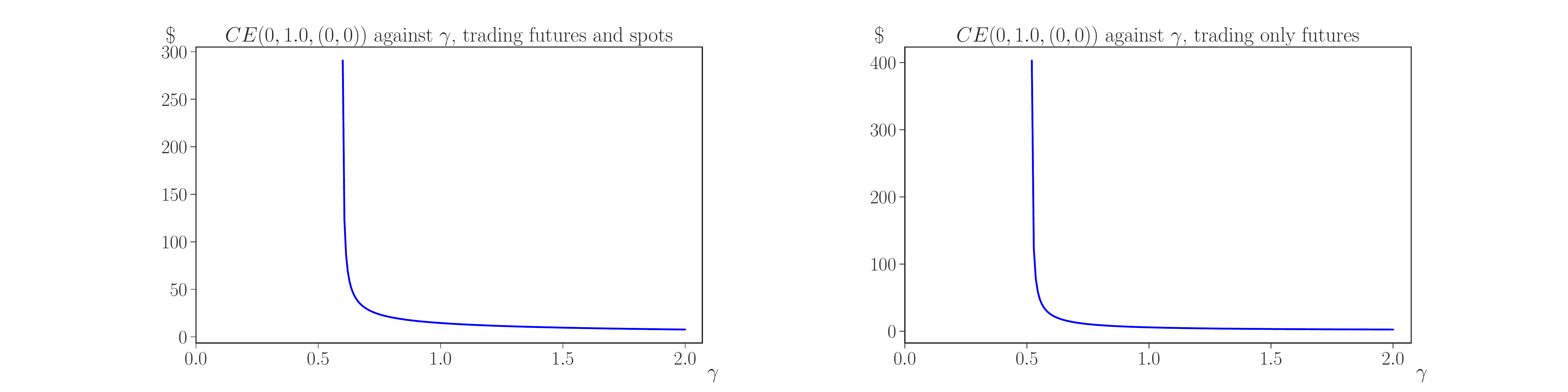}}
% }
}
\caption{Sensitivity of $CE(t=0,x=1.0,\zv=(0,0))$ with respect to the risk aversion parameter $\gam$. Note that we are considering that case $0<\gam<1$ despite the fact that our assumption explicitly excludes this case. The left (resp. right) plot corresponds to trading both futures and the underlying assets (resp. only futures). As can be seen in both cases, the certainty equivalent becomes arbitrary large as $\gam$ approaches certain critical value in the interval $(0,1)$. This indicates existence of nirvana strategies for some values of $\gam\in(0,1)$.
\vspace{1em}
\label{fig:gam_sens}}
\end{figure}

\section{Sensitivity analysis and a simulated example}\label{sec:example}

In this section, we provide a few numerical illustration to investigate how the optimal trading strategies and the CE values are affected by various model parameters. We also include a simulated example to see trading in action.

Let us first consider the effect of risk aversion parameter $\gam$. Since the ``shifted'' utility function $U(x)=(x^{1-\gam}-1)/(1-\gam)$ is decreasing in $\gam$, and shifting the utility function does not change the CE values, one expects that CE values to be decreasing in $\gam$. Figure \ref{fig:gam_sens} confirms this behavior. It shows the CE values at $(t,x,z)=(0,1.0,0)$ as a function of $\gam$, for both cases when the underlying is traded and when it is not. Note that we have included the values of $\gam\in(0,1)$, although we have explicitly excluded these values from our analysis in Sections \ref{sec:Futures} and \ref{sec:SpotFutures}. The plots shows that the CE values approach infinity for certain level of risk-aversion in the range $(0,1)$. As we have mentioned already, this behavior is consistent with the findings of \cite{AngoshtariLeung2019} and indicates the existence of the so-called \emph{nirvana} strategies.

Next, we consider the effect of the mean reversion rates, that is, the parameters $\kap_i:=\eta_{i,S}-\eta_{i,F}$ in \eqref{eq:Bb}. One expects higher mean-reversion to lead to higher CE values, as deviation from equilibrium states are more certain to be corrected. Figure \ref{fig:kappa_sens} confirms this effect. It shows CE values at $(t,x,z)=(0,1.0,0)$ for various multiples of $\kap_1$ and $\kap_2$. In particular, for a given value of $\kap\in(0,1.5)$, it is assumed that $(\eta_{1,F},\eta_{2,F})=(-\kap, -1.5\kap)$ and $\eta_{1,S}=\eta_{2,S}=0$ such that the mean-reversion rates in \eqref{eq:Bb} are given by $\kap_1=\kap$ and $\kap_2=1.5\kap$. Note that CE values are much more affected by the mean-reversion rate when the underlying assets are traded compared to when they are not traded. Note, also, that as $\kap\to0^+$, CE values approach almost 1 (the exact values are around 1.004). This indicates that most of the profitability originates from the mean-reversion (as apposed to the drift in the prices).

\begin{figure}[p]
\centerline{
% \fbox{
\adjustbox{trim={0.0\width} {0.0\height} {0.0\width} {0.0\height},clip}
{\includegraphics[scale=0.3, page=1]{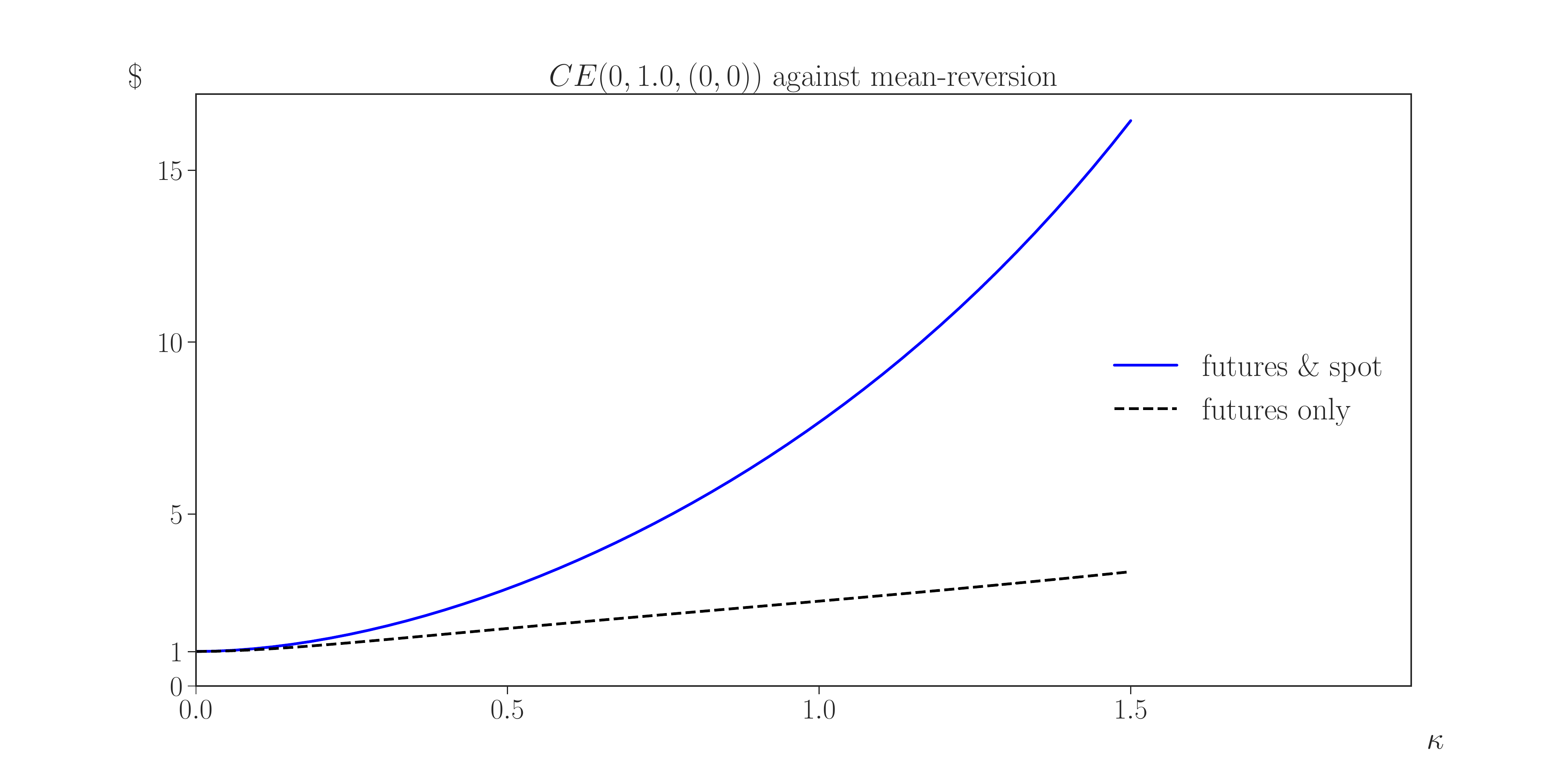}}
% }
}
% \centerline{
% % \fbox{
% \adjustbox{trim={0.0\width} {0.0\height} {0.0\width} {0.0\height},clip}
% {\includegraphics[scale=0.23, page=1]{kappa_sens_near_zero.pdf}}
% % }
% }
\caption{Sensitivity of $CE(t=0,x=1.0,\zv=(0,0))$ with respect to the mean-reversion rate of the log-bases. In particular, for each value of $\kap$, it is assumed that $(\eta_{1,F},\eta_{2,F})=(-\kap, -1.5\kap)$ and $\eta_{1,S}=\eta_{2,S}=0$ such that the mean-reversion rates in \eqref{eq:Bb} are given by $\kap_1=\kap$ and $\kap_2=1.5\kap$. The solid blue (resp. the dashed black) curve corresponds to trading both futures and the underlying assets (resp. only futures). The CE values increase as the mean-reversion rate increases.
% \red{[Fix the mesh issue for the top right plot.]}
\vspace{1em}
\label{fig:kappa_sens}}
\end{figure}
\vspace{5pt}

% \begin{figure}[t]
% \centerline{
% % \fbox{
% \adjustbox{trim={0.0\width} {0.0\height} {0.0\width} {0.0\height},clip}
% {\includegraphics[scale=0.23, page=1]{futures_cor_sens.pdf}}
% % }
% }
% \caption{Sensitivity of $CE(t=0,x=1.0,\zv=(0,0))$ to $\sig_{12}$, i.e. the instantaneous correlation between the futures $\sig_{12}$. The left (resp. right) plot corresponds to trading both futures and the underlying assets (resp. only futures). As can be seen higher positive correlation significantly increase CE values when trading both the futures and the underlying. When trading only futures, however, higher correlation size lower the CE values and the magnitude of change is small.
% \vspace{1em}
% \label{fig:futures_cor_sens}}
% \end{figure}
%
%
% \begin{figure}[t]
% \centerline{
% % \fbox{
% \adjustbox{trim={0.0\width} {0.0\height} {0.0\width} {0.0\height},clip}
% {\includegraphics[scale=0.23, page=1]{spot_cor_sens.pdf}}
% % }
% }
% \caption{Sensitivity of $CE(t=0,x=1.0,\zv=(0,0))$ to $\sig_{34}$, namely, the instantaneous correlation between the underlying assets. The left (resp. right) column of plots corresponds to trading both futures and the underlying assets (resp. only futures). As can be seen, when trading only the futures, higher positive correlation between the underlying assets reduces CE values. When trading both futures and the underlying assets, however, the CE values first decrease and then increases.
% \vspace{1em}
% \label{fig:spot_cor_sens}}
% \end{figure}

\begin{figure}[p]
\centerline{
% \fbox{
\adjustbox{trim={0.0\width} {0.0\height} {0.0\width} {0.0\height},clip}
{\includegraphics[scale=0.3, page=1]{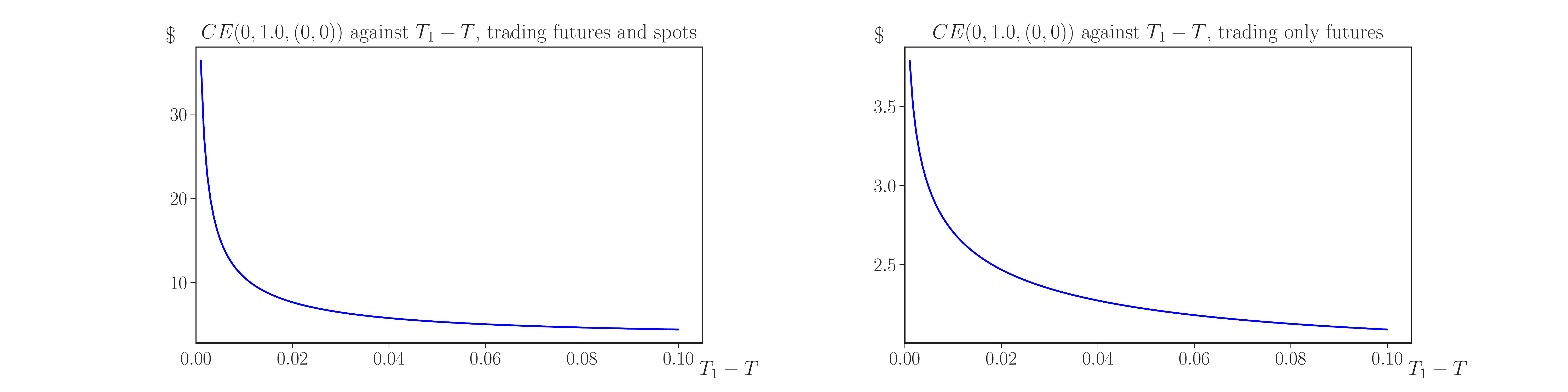}}
% }
}
\caption{Sensitivity of $CE(t=0,x=1.0,\zv=(0,0))$ to $T_1$, i.e. the maturity of the first futures contract. The left (resp. right) column of plots corresponds to trading both futures and the underlying assets (resp. only futures). CE values are decreasing in $T_1$. In particular, in the limit $T_1\to T^+$, the market includes arbitrage opportunities and CE values approach infinity.\hspace{\textwidth}
\emph{Parameters:} $\gam=2.0$. The remaining parameters (other than $T_1$) are as in Figure \ref{fig:sim-FSZ}.
\vspace{1em}
\label{fig:eps_sens}}
\end{figure}

Finally, let us consider the effect of the delivery time of the futures contracts. One expects the delivery dates that are nearer to the end of the trading horizon, to result in higher CE values. This is because lower values of $T_i-T$ mean stronger convergence of futures and spot at $T$. Indeed, as we have mentioned in Remark \ref{rem:arbitrage}, our market model includes arbitrage in the limit $T\to T_i^-$ and CE values should approach infinity at this limit. These observations are confirmed by Figure \ref{fig:eps_sens}. It shows the CE values as a function of $T_1-T$, that is, the time to maturity of the first futures contract at the end of the trading horizon $T$. As one can see, the CE values increases to $+\infty$ as $T\to T_i^-$.

\begin{figure}[p]
\centerline{
% \fbox{
\adjustbox{trim={0.0\width} {0.0\height} {0.0\width} {0.0\height},clip}
{\includegraphics[scale=0.23, page=1]{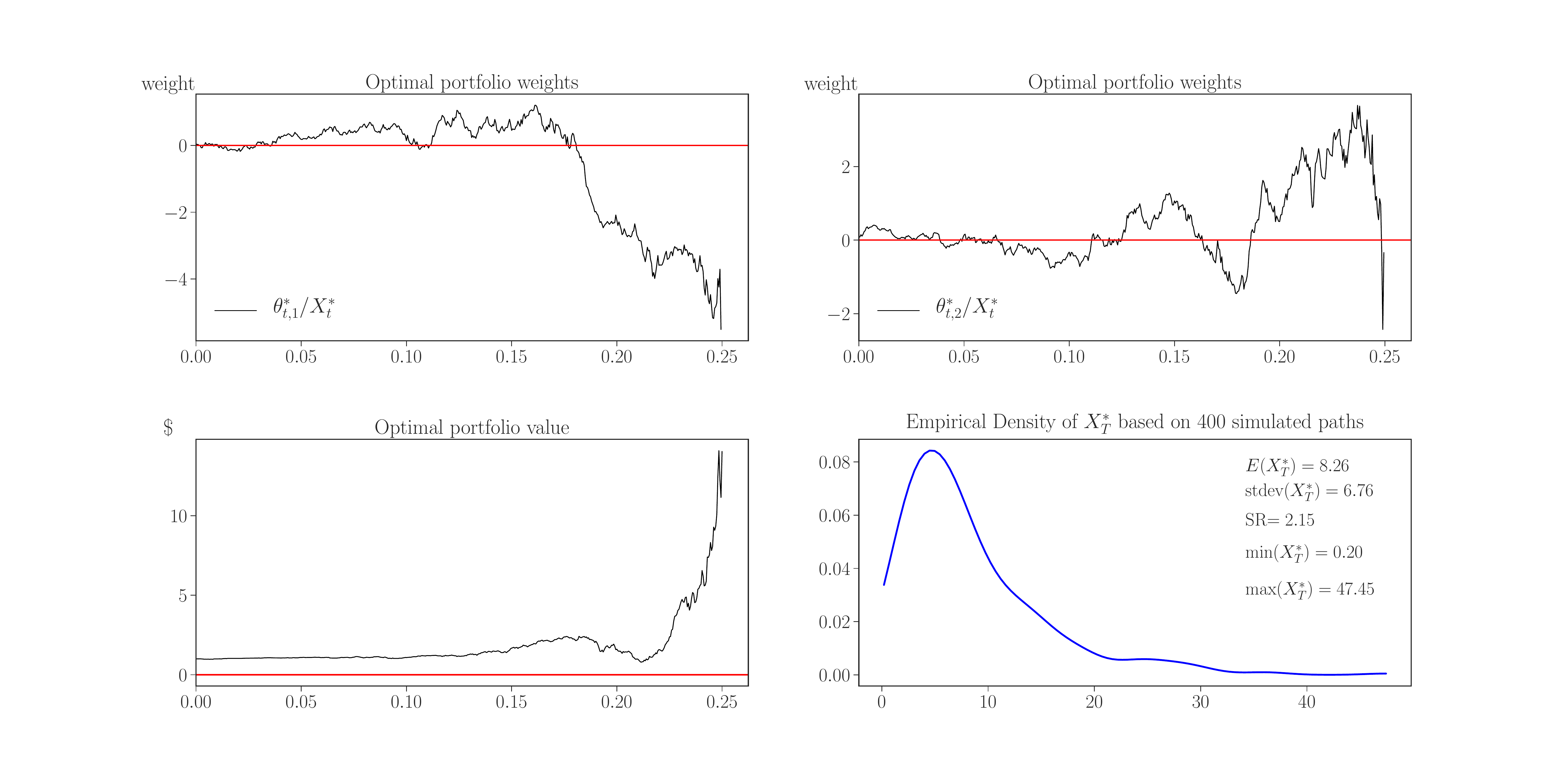}}
% }
}
\caption{\emph{Top:} Positions for the optimal futures only trading strategy assuming the  price paths in the top and middle rows of Figure \ref{fig:sim-FSZ}. \emph{Bottom left:} The corresponding path of the portfolio value. \emph{Bottom right:} The kernel density estimation (KDE) of the optimal terminal wealth $X^*_T$ based on 400 simulated price paths, all of which started with the initial wealth $X^*_0=1$. \vspace{1em}
\label{fig:simTrading-F}}
\vspace{1em}
\end{figure}

\begin{figure}[p]
\centerline{
% \fbox{
\adjustbox{trim={0.0\width} {0.0\height} {0.0\width} {0.0\height},clip}
{\includegraphics[scale=0.23, page=1]{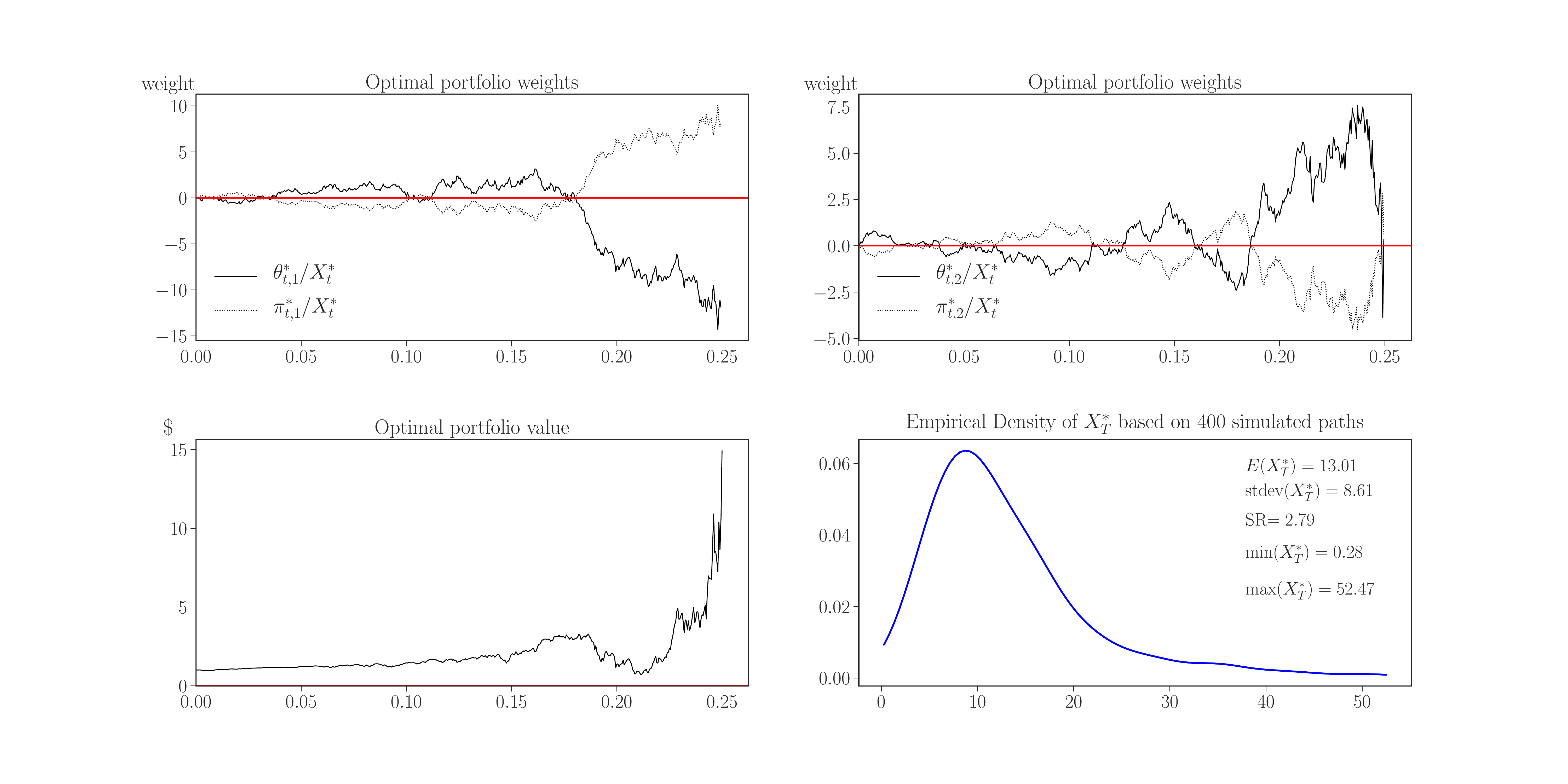}}
% }
}
\caption{\emph{Top:} Positions for the optimal futures and spots trading strategy assuming the  price paths in the top and middle rows of Figure \ref{fig:sim-FSZ}. \emph{Bottom left:} The corresponding path of the portfolio value. \emph{Bottom right:} The KDE of the optimal terminal wealth $X^*_T$ based on 400 simulated price paths, all of which started with the initial wealth $X^*_0=1$. \vspace{1em}
\label{fig:simTrading-FS}}
\vspace{1em}
\end{figure}

We end this section by a simulated example. The market is simulated $400$ times according to the discretization scheme of Corollary \ref{cor:Sim} using market parameters from the caption of Figure \ref{fig:sim-FSZ}. For each simulated market path, we trade once according to the strategy given by Figures \ref{fig:optimF_theta} (i.e. trading  futures only), and another time according to the strategy given by \ref{fig:optimFS_theta} (i.e. trading both the futures and the underlying assets). In each case, the initial portfolio value is set to be \$1, i.e. $x=1$. As a result, 400 paths of the portfolio values are obtained. In Figures \ref{fig:simTrading-F} and \ref{fig:simTrading-FS}, we illustrate  the positions and the portfolio value corresponding to the price paths from  Figure \ref{fig:sim-FSZ}. We have also included the kernel density estimation (KDE) of the density of $X^*_T$, i.e. the optimal portfolio value at $T$, along with a few summary statistics. The density appears to be right skewed for both trading strategies. The Sharpe ratio of the futures-only strategy is 2.15, which is lower than 2.79 achieved by the futures-and-underlying strategy. Finally, both strategies may result in losses. The minimum terminal wealth for the futures-only strategy is 0.2, while the worst terminal wealth for the futures-and-underlying strategy is 0.28.

\section{Concluding Remarks}\label{sec:conclude}

We have introduced a new stochastic model for the joint dynamics of bases among different futures with different spot assets.  It    captures  not only the dependency structure among futures and spot assets but also the path behaviors of the associated bases. With {continuous-time} rebalancing throughout the trading horizon,  the optimal trading problem admits a solution semi-explicitly in terms of the solution of a matrix Riccati differential equation. This allows for instant computation and analysis of the optimal trading strategies. Future research directions include pricing and hedging commodity derivatives under the stochastic basis model  and studying dynamic futures trading under other risk preferences. Alternative approaches to trading futures, such as rolling and timing strategies,\footnote{See \cite{LeungXin2016} and \cite{LeungLiLiZheng2015} for   discussions of such strategies involving a single futures contract.} can also be explored using this  model.

\clearpage

%-------------------------------------------------------
%
%					Bibliography
%
%-------------------------------------------------------
\pagebreak
\bibliographystyle{chicago}
\bibliography{references}

%-----------------------------------------------------------------------------------
%
%       SECTION: 		Proof of Theorem 1
%
%-----------------------------------------------------------------------------------
\appendix

\section{Proof of Theorem \ref{thm:Futures}}\label{app:Futures}

The proof relies on the following well-known comparison result for Riccati differential equations, which we include for readers' convenience. Let $A\ge0$ (resp. $A>0$) denote that A is positive semi-definite (resp. positive definite) and $A\ge B$ (resp. $A>B$) denote that $A-B\ge0$ (resp. $A-B>0$). 

\begin{lemma}[\cite{Reid1972}, Theorem 4.3, p.~122]\label{lem:comp}
	Let $A(t)$, $B(t)$, $C(t)$, and $\Ct(t)$ be continuous $N\times N$ matrix functions on an interval $[a,b]\in\Rb$, and $H_0$ and $\Ht_0$ be two symmetric $N\times N$ matrices. Furthermore, assume that for all $t\in[a,b]$, $B(t)\ge0$ and $C(t)$ and $\Ct(t)$ are symmetric such that $C(t)\ge\Ct(t)$. Consider the Riccati matrix differential equations
	\begin{align}\label{eq:Riccati1}
	\begin{cases}
		H'(t) + H(t)\,B(t)\, H(t) + H(t)A(t) + A(t)^\top H(t) - C(t) =0;
		\quad a\le t \le b,\\
		H(a) = H_0, 
	\end{cases}
	\end{align}
	and
	\begin{align}\label{eq:Riccati2}
	\begin{cases}
		\Ht'(t) + \Ht(t)\,B(t)\, \Ht(t) + \Ht(t)A(t) + A(t)^\top \Ht(t) - \Ct(t) =0;
		\quad a\le t \le b,\\
		H(a) = \Ht_0.
	\end{cases}
	\end{align}
	If $H_0>\Ht_0$ (resp. $H_0\ge\Ht_0$) and \eqref{eq:Riccati2} has a symmetric solution $\Ht(t)$, then \eqref{eq:Riccati1} also has a symmetric solution $H(t)$ such that $H(t)>\Ht(t)$ (resp. $H(t)\ge\Ht(t)$) for all $t\in[a,b]$. 
\end{lemma}\vspace{1em}

\noindent{\bf Proof of (i):}			
Consider the matrix Riccati differential equation
\begin{align}
\begin{cases}
	\Ht'(\tau) 
	+ \Ht(\tau)\, \big(\gam\,C^\top\Sig\,C + (1-\gam)A\big)\, \Ht(\tau)\\*
	\hspace{3em}- 2\left(\etav(T-\tau)^\top C + \left(\frac{1}{\gam}-1\right)\etav_F(T-\tau)^\top B\right)\, \Ht(\tau)=0_{N\times N},\\
	H(0)=0_{N\times N},
\end{cases}
\end{align}
which has the trivial solution $\Ht_0\equiv 0_{N\times N}$. Assume, for now, that $\gam\,C^\top\Sig\,C + (1-\gam)A >0$. By Lemma \ref{lem:comp}, it follows that \eqref{eq:H-ODE-F} has a positive semidefinite solution $H(\tau)$ on $[0,T]$. That $H(\tau)$ is positive definite on $(0,T]$ follows from the fact that $H'(0) = \frac{\gam-1}{\gam^2} \etav_F(T-\tau)^\top \Sig_\Fv^{-1}\etav_F(T-\tau)>0$ and Lemma \ref{lem:comp}. The uniqueness of the solution follows from the uniqueness theorem for first order differential equations.
		
	It only remains to show that $\gam\,C^\top\Sig\,C + (1-\gam)A >0$. By \eqref{eq:Cov}, \eqref{eq:A}, and \eqref{eq:EClam}, we have
	\begin{align}
		A &= \Sig_\Fv + \Sig_\Sv - \Sig_{\Fv\Sv} - \Sig_{\Fv\Sv}^\top - \Sig_\Sv + \Sig_{\Fv\Sv}^\top\Sig_\Fv^{-1}\Sig_{\Fv\Sv}\\*
		&= C^\top \Sig C - \left(\Sig_\Sv - \Sig_{\Fv\Sv}^\top\Sig_\Fv^{-1}\Sig_{\Fv\Sv}\right).
	\end{align}
	From \eqref{eq:Cov} and \eqref{eq:Sigt}, we obtain
	\begin{align}
		\Sig_\Sv - \Sig_{\Fv\Sv}^\top \Sig_\Fv^{-1}\Sig_{\Fv\Sv} &=
		\Sigt_\Sv\Sigt_\Sv^\top + \Sigt_{\Fv\Sv}^\top\Sigt_{\Fv\Sv}
		- \Sigt_{\Fv\Sv}^\top \Sigt_\Fv^\top 
		\left(\Sigt_{\Fv}\Sigt_{\Fv}^\top\right)^{-1}
		\Sigt_\Fv\Sigt_{\Fv\Sv}\\*
		&=\Sigt_\Sv\Sigt_\Sv^\top
	\end{align}
	Finally, using the assumption $\gam>1$ and the last two results yield
	\begin{align}
		\gam\,C^\top\Sig\,C + (1-\gam)A &= C^\top\Sig\,C + (\gam-1)\left(\Sig_\Sv - \Sig_{\Fv\Sv}^\top \Sig_\Fv^{-1}\Sig_{\Fv\Sv}\right)\\*
		&= C^\top\Sig\,C + (\gam -1) \Sigt_\Sv\Sigt_\Sv^\top >0,
	\end{align}
	as we set out to prove.
\vspace{1em}

	\noindent{\bf Proof of (ii) and (iii):} As we argue later, $\VF(t,x,\zv)$ is the solution of the Hamilton-Jacobi-Bellman (HJB) equation
	\begin{align}\label{eq:HJB-F}
	\begin{cases}
		\displaystyle
		\vp_t + \sup_{\thtv\in\Rb^N} \Jc_\thtv \vp = 0,\\
		\vp(T,x,\zv) = \frac{x^{1-\gam}}{1-\gam},
	\end{cases}
	\end{align}
	for $(t,x,\zv)\in[0,T]\times\Rb_+\times\Rb^N$, in which the differential operator $\Jc_\thtv$ is given by
	\begin{align}
		\Jc_\thtv \varphi(t, x,\zv) :={}
		&(\mv+C^\top \etav(t)\zv)^\top \varphi_\zv
		+ \frac{1}{2}\tr(C^\top\Sig\,C\,\varphi_{\zv\zv})\\*
		&{}+ \big(r x + \thtv^\top(\muv_\Fv+\etav_F(t)\zv)\big)\,\varphi_x\\*
		&{}+ \frac{1}{2} \thtv^\top\Sig_\Fv\,\thtv \,\varphi_{xx}
		+ \thtv^\top (\Sig_\Fv - \Sig_{\Fv\Sv})\, \varphi_{x\zv},
	\end{align}
	for any $\thtv\in\Rb^N$ and any $\varphi(t,x,\zv):[0,T]\times\Rb_+\times\Rb^N\to\Rb$ that is continuously twice differentiable in $(x,\zv)$. Here, $\tr(A)$ is the trace of the matrix A and we have used the shorthand notations $\varphi_\zv$ and $\varphi_{\zv\zv}$ to denote the gradient vector and the Hessian matrix of $\varphi(t,x,\zv)$ with respect to $\zv$, that is,
	\begin{align}
		\varphi_\zv(t,x,\zv) := \left(\frac{\partial \varphi}{\partial z_1},\dots,
			\frac{\partial \varphi}{\partial z_N}\right),\quad\text{and}\quad
		\varphi_{\zv\zv}(t,x,\zv) := \left[\frac{\partial^2 \varphi}{\partial z_i\partial z_j}\right]_{N\times N}.
	\end{align}

	By assuming that $\vp_{xx}<0$ (which is verified by the form of the solution, i.e. \eqref{eq:ansatz}), the maximizer of the left side of the differential equation in \eqref{eq:HJB-F} is
	% \begin{align}
	% \begin{split}
	% 	\Jc_\thtv \vp(t,x,\zv) ={}
	% 	&(\mv+C^\top E(t)\zv)^\top \vp_\zv
	% 	+ \frac{1}{2}\tr(C^\top\Sig\,C\,\vp_{\zv\zv})
	% 	+ rx\, \vp_x\\*
	% 	&{}-\frac{1}{2}(\alv+E_\Fv(t)\zv)^\top\Sig_\Fv^{-1}(\al+E_\Fv(t)\zv)\frac{\vp_x^2}{\vp_{xx}}\\*
	% 	&{}-\frac{1}{2\vp_{xx}}\vp_{x\zv}^\top \big(\Sig_\Fv + \Sig_{\Sv\Fv}\Sig_\Fv^{-1}\Sig_{\Fv\Sv} - \Sig_{\Sv\Fv}-\Sig_{\Fv\Sv}\big)\,\vp_{x\zv}
	% 	-\frac{\vp_x}{\vp_{xx}}(\alv+E(t)\zv)^\top (I_N - \Sig_\Fv^{-1}\Sig_{\Fv\Sv}) \vp_{x\zv}\\*
	% 	&{}+\frac{1}{2}\vp_{xx}\left\|
	% 								\Sigt_\Fv^\top \thtv
	% 								+ \frac{\vp_x}{\vp_{xx}}\Sigt_\Fv^{-1}(\alv+E_\Fv(t)\zv)
	% 								+ \frac{1}{\vp_{xx}} \Sigt_\Fv^{-1}(\Sig_\Fv-\Sig_{\Fv\Sv})\,\vp_{x\zv}
	% 						   \right\|^2
	% \end{split}
	% \end{align}
	\begin{align}\label{eq:CandidateOptimal-F}
		\thtv^*(t,x,\zv) =
		- \frac{\vp_{x}(t,x,\zv)}{\vp_{xx}(t,x,\zv)} \Sig_\Fv^{-1}(\muv_\Fv+\etav_F(t)\zv)
		- \BF\,\frac{\vp_{x\zv}(t,x,\zv)}{\vp_{xx}(t,x,\zv)}.
	\end{align}
	Substituting $\sup_{\thtv\in\Rb^N} \Jc_\thtv \vp=\Jc_{\thtv^*}\vp$ in \eqref{eq:HJB-F} yields
	\begin{align}\label{eq:HJB2-F}
		\vp_t &{}+ (\mv+C^\top \etav(t)\zv)^\top \vp_\zv
		+ \frac{1}{2}\tr(C^\top\Sig\,C\,\vp_{\zv\zv})
		+ rx\, \vp_x\\*
		&{}-\frac{1}{2}\big(\muv_\Fv+\etav_F(t)\zv\big)^\top\Sig_\Fv^{-1}\big(\muv_\Fv+\etav_F(t)\zv\big)\frac{\vp_x^2}{\vp_{xx}}\\*
		&{}-\frac{1}{2\vp_{xx}}\vp_{x\zv}^\top \AF\,\vp_{x\zv}
		-\frac{\vp_x}{\vp_{xx}}(\muv_\Fv+\etav_F(t)\zv)^\top \BF\, \vp_{x\zv}
		= 0,
	\end{align}
	for $(t,x,\zv)\in[0,T]\times\Rb_+\times\Rb^N$, subject to the terminal condition
	\begin{align}\label{eq:TC-F}
		\vp(T,x,\zv) = \frac{x^{1-\gam}}{1-\gam}.
	\end{align}
	To solve \eqref{eq:HJB2-F}, we consider the ansatz
	\begin{align}\label{eq:ansatz}
		\vp(t,x,\zv) &=
		\frac{x^{1-\gam}}{1-\gam}
		\ee^{\gam\left(f(T-t) + \zv^\top \gv(T-t) - \frac{1}{2} \zv^\top H(T-t)\zv\right)},
		% \\*
		% &=
		% \frac{x^{1-\gam}}{1-\gam}
		% \exp\left(f(t) +
		% \sum_{i=1}^{2} g_i(t)\, z_i + \frac{1}{2}\sum_{i,j=1}^2  h_{ij}(t)\, z_i\,z_j
		% \right),
	\end{align}
	for $(t,x,\zv)\in[0,T]\times\Rb_+\times\Rb^N$, in which $f(t)$, $\gv(t)=\big(g_1(t),\ldots,g_N(t)\big)^\top$, and
	\begin{align}
		H(t) =
		\begin{pmatrix}
			h_{11}(t) & \dots & h_{1N}(t)\\
			\vdots & \ddots & \vdots\\
			h_{N1}(t) & \dots & h_{NN}(t)
		\end{pmatrix} 
	\end{align}
	are unknown functions to be determined. Without loss of generality, we further assume that $H(t)$ is symmetric such that $h_{ij}(t)=h_{ji}(t)$ for $0\le t\le T$. Substituting this ansatz into \eqref{eq:HJB2-F} 	yields
	\begin{align}
		&\frac{\gam}{2}\zv^\top\bigg[
			H' 
			+ H\, \big(\gam\,C^\top\Sig\,C + (1-\gam)A\big)\, H\\*
		&\qquad\qquad{}- 2\left(E^\top C + \left(\frac{1}{\gam}-1\right)\etav_F^\top B\right)\, H
			+ \frac{1-\gam}{\gam^2} \etav_F^\top \Sig_\Fv^{-1}\etav_F
		\bigg]\zv\\*
		&{}+\gam\zv^\top\bigg[
			-\gv'
			+ \left(E^\top C - \gam\,H\,C^\top\Sig\,C - (1-\gam)HA 
			+ \left(\frac{1}{\gam}-1\right)\etav_F^\top B\right)\gv\\
			&\hspace{4em}{}- H\left(\mv + \left(\frac{1}{\gam}-1\right) B^\top\muv_\Fv\right)
			+\frac{1-\gam}{\gam^2}\,\etav_F^\top\Sig_\Fv^{-1}\muv_\Fv
		\bigg]\\*
		&{}- \gam f'
		+(1-\gam)\left(r+\frac{\muv_\Fv^\top\Sig_\Fv^{-1}\muv_\Fv}{2\gam}\right)
		+\frac{\gam}{2}\gv^\top\left((1-\gam) A + \gam\, C^\top \Sig\,C\right)\gv\\
		&{}+ \gam\left(\mv^\top+\left(\frac{1}{\gam}-1\right)\muv_\Fv^\top B\right)\gv
		-\frac{\gam}{2}\tr(C^\top\Sig\,C\, H)
		=0,
	\end{align}
	for all $(t,\zv)\in[0,T]\times\Rb^N$, where we have omitted the $t$ arguments to simplify the terms. Taking the terminal condition \eqref{eq:TC-F} into account, it then follows that $H$, $\gv$, and $f$ must satisfy \eqref{eq:H-ODE-F}, \eqref{eq:g-ODE-F}, and \eqref{eq:f-integral-F}, respectively.
	
	By statement (i) of the theorem, \eqref{eq:H-ODE-F} has a unique solution that is positive definite on $(0,T]$. Using the classical existence theorem of systems of ordinary differential equations, we then deduce that \eqref{eq:g-ODE-F} also has a unique bounded solution on $[0,T]$. Finally, $f$ given by \eqref{eq:f-integral-F} is continuously differentiable since the integrand on the right side is continuous. Thus, $v(t,x,\zv)$ given by \eqref{eq:VF-sol-F} is a solution of the HJB equation \eqref{eq:HJB-F}.

	It only remains to show that the solution of the HJB equation is the value function, that is $\vp(t,x,\zv)=V(t,x,\zv)$ for all $(t,x,\zv)\in[0,T]\times\Rb_+\times\Rb^N$. Note that, for $0\le t <T$, we have
	\begin{align}
		f(T-t) &{}+ \zv^\top \gv(T-t) - \frac{1}{2} \zv^\top H(T-t)\zv\\*
		&= f(T-t) + \frac{1}{2} \gv(T-t)^\top H(T-t)\,\gv(T-t)
		-\frac{1}{2}\left\|\Ht^\top \zv - \Ht^{-1}\gv(T-t)\right\|^2\\*
		&\le f(T-t) + \frac{1}{2} \gv(T-t)^\top H(T-t)\,\gv(T-t),
	\end{align}
	in which $\Ht(T-t)$ is the Cholesky factor of the positive definite matrix $H(T-t)$. It then follows that $\vp(t,x,\zv)$ in \eqref{eq:ansatz} is bounded in $\zv$ and has polynomial growth in $x$. A standard verification result such as Theorem 3.8.1 on page 135 of \cite{FlemingSoner2006} then yields that $\vp(t,x,\zv)=V(t,x,\zv)$ for all $(t,x,\zv)\in[0,T]\times\Rb_+\times\Rb^N$.
	
	The verification result also states that the optimal control in feedback form is $\thtv^*$ given by \eqref{eq:CandidateOptimal-F}. Using \eqref{eq:ansatz}, one obtains $\thtv^*(t,x,\zv)$ in terms $H$ and $\gv$ as in \eqref{eq:OptimalStrat-F}.

%-----------------------------------------------------------------------------------
%
%       SECTION: Proof of Theorem 2
%
%-----------------------------------------------------------------------------------
\section{Proof of Theorem \ref{thm:VF}}\label{app:VF}

	The proof is similar to the proof of Theorem \ref{thm:Futures} and, thus, is presented in less detail.\vspace{1em}
	
	\noindent{\bf (i):} Similar to the proof of statement (i) of Theorem \ref{thm:Futures}, the proof here involves comparing \eqref{eq:H-ODE} with the homogenous equation
		\begin{align}
		\begin{cases}
			\Ht_0'(\tau) +
			\Ht_0(\tau) C^\top \Sig\, C\, \Ht_0(\tau)
			- \frac{2}{\gam}\,\etav(T-\tau)^\top C\, \Ht_0(\tau)=0;
			\quad0\le \tau\le T,\\
			\Ht_0(0)=0_{N\times N},
		\end{cases}
		\end{align}
		using Lemma \ref{lem:comp}.
	\vspace{1em}

	\noindent{\bf (ii) and (iii):} As we later verify, $V(t,x,\zv)$ solves the HJB equation
	\begin{align}\label{eq:HJB}
	\begin{cases}
		\displaystyle
		\vp_t + \sup_{\Thtv\in\Rb^{2N}} \Lc_\Thtv \vp = 0,\\
		\vp(T,x,\zv) = \frac{x^{1-\gam}}{1-\gam},
	\end{cases}
	\end{align}
	for $(t,x,\zv)\in[0,T]\times\Rb_+\times\Rb^N$, in which the differential operator $\Lc_\Thtv$ is given by
	\begin{align}
		\Lc_\Thtv \varphi(t, x,\zv) :={}
		&(\mv+C^\top \etav(t)\zv)^\top \varphi_\zv
		+ \frac{1}{2}\tr(C^\top\Sig\,C\,\varphi_{\zv\zv})\\*
		&{}+ \big(r x + \Thtv^\top(\muv+\etav(t)\zv)\big)\varphi_x
		+ \frac{1}{2} \Thtv^\top\Sig\,\Thtv\,\varphi_{xx}
		+ \Thtv^\top \Sig\, C\, \varphi_{x\zv},
	\end{align}
	for any $\Thtv\in\Rb^{2N}$ and any $\varphi(t,x,\zv):[0,T]\times\Rb_+\times\Rb^N\to\Rb$ that is continuously twice differentiable in $(x,\zv)$.
	Assuming that $\vp_{xx}<0$, which is readily verified by the form of the solution in \eqref{eq:VF-sol}, we obtain that the maximizer $\Thtv^*$ in \eqref{eq:HJB} is given by 
	\begin{align}\label{eq:CandidateOptimal}
		\Thtv^*
		=\begin{pmatrix}
			\thtv^*(t,x,\zv)\\
			\piv^*(t,x,\zv)
		\end{pmatrix}
		=
		- \frac{\vp_{x}(t,x,\zv)}{\vp_{xx}(t,x,\zv)} \Sig^{-1} (\muv+\etav(t)\zv)
		- C\,\frac{\vp_{x\zv}(t,x,\zv)}{\vp_{xx}(t,x,\zv)}.
	\end{align}
	Substituting $\Thtv^*$ into \eqref{eq:HJB} yields
	\begin{align}\label{eq:HJB2}
		\vp_t &{}+ (\mv+C^\top \etav(t)\zv)^\top \vp_\zv
		+ \frac{1}{2}\tr(C^\top\Sig\,C\,\vp_{\zv\zv})
		+ rx\, \vp_x\\*
		&{}-\frac{1}{2}(\muv+\etav(t)\zv)^\top\Sig^{-1}(\muv+\etav(t)\zv)\frac{\vp_x^2}{\vp_{xx}}\\*
		&{}-\frac{1}{2\vp_{xx}}\vp_{x\zv}^\top C^\top \Sig\,C\,\vp_{x\zv}
		-\frac{\vp_x}{\vp_{xx}}(\muv+\etav(t)\zv)^\top C\, \vp_{x\zv}
		= 0,
	\end{align}
	for $(t,x,\zv)\in[0,T]\times\Rb_+\times\Rb^N$, subject to the terminal condition
	\begin{align}\label{eq:TC}
		\vp(T,x,\zv) = \frac{x^{1-\gam}}{1-\gam}.
	\end{align}
	This partial differential equation is similar to \eqref{eq:HJB2-F} and can be solved using the same ansatz.
	% \begin{align}
% 		\vp(t,x,\zv) &=
% 		\frac{x^{1-\gam}}{1-\gam}
% 		\ee^{\gam\left(f(T-t) + \zv^\top \gv(T-t) - \frac{1}{2} \zv^\top H(T-t)\zv\right)},
% 		% \\*
% 		% &=
% 		% \frac{x^{1-\gam}}{1-\gam}
% 		% \exp\left(f(t) +
% 		% \sum_{i=1}^{2} g_i(t)\, z_i + \frac{1}{2}\sum_{i,j=1}^2  h_{ij}(t)\, z_i\,z_j
% 		% \right),
% 	\end{align}
% 	for $(t,x,\zv)\in[0,T]\times\Rb_+\times\Rb^N$, in which $f(t)$, $\gv(t)=\big(g_1(t),\ldots,g_N(t)\big)^\top$, and
% 	\begin{align}
% 		H(t) =
% 		\begin{pmatrix}
% 			h_{11}(t) & \dots & h_{1N}(t)\\
% 			\vdots & \ddots & \vdots\\
% 			h_{N1}(t) & \dots & h_{NN}(t)
% 		\end{pmatrix}
% 	\end{align}
% 	are unknown functions to be determined. Without loss of generality, we further assume that $H(t)$ is symmetric such that $h_{ij}(t)=h_{ji}(t)$ for $0\le t\le T$. Substituting this ansatz into \eqref{eq:HJB2} 	yields
Indeed, applying \eqref{eq:ansatz} yields that $f(t)$, $\gv(t)$, and $H(t)$ satisfy
	% and noting that
	% \begin{align}
	% 	\vp_\zv& = \vp (\gv+H\zv),
	% \intertext{and}
	% 	\tr(C^\top\Sig\,C\,\vp_{\zv\zv}) &= \vp\tr\left(C^\top\Sig\,C\,(H+\gv\gv^\top+H\zv\zv^\top H + H \zv \gv^\top + \gv \zv^\top H)\right)\\*
	% 	&=\vp\left[\tr\left(C^\top\Sig\,C\,(H+\gv\gv^\top)\right)
	% 	+\tr\left(C^\top\Sig\,C\,H\zv\zv^\top H\right)
	% 	+\tr\left(C^\top\Sig\,C\,(H \zv \gv^\top + \gv \zv^\top H)\right)
	% 	\right]\\*
	% 	&=\vp\left[
	% 	\tr\left(C^\top\Sig\,C\,H\right)
	% 	+\gv^\top C^\top\Sig\,C\,\gv
	% 	+\zv^\top H\,C^\top\Sig\,C\,H\zv
	% 	+ 2\,\zv^\top H\,C^\top\Sig\,C\,\gv
	% 	\right],
	% \end{align}
	\begin{align}
		&\frac{\gam}{2}\zv^\top\left[
			H' 
			+ H\, C^\top\Sig\,C\, H
			- \frac{2}{\gam} E^\top C\, H
			+ \frac{1-\gam}{\gam^2} E^\top \Sig^{-1}E
		\right]\zv\\*
		&{}+\gam\zv^\top\bigg[
			-\gv'
			+ (\frac{1}{\gam} E^\top-H\,C^\top\Sig)\,C\,\gv\\*
		&\hspace{3.5em}{}- H\left(\mv + \left(\frac{1}{\gam}-1\right) C^\top\muv\right)
			+\frac{1-\gam}{\gam^2}\,E^\top\Sig^{-1}\muv
		\bigg]\\*
		&{}- \gam f'
		+(1-\gam)\left(r+\frac{\muv^\top\Sig^{-1}\muv}{2\gam}\right)
		+\frac{\gam}{2}\gv^\top C^\top \Sig\,C\,\gv\\*
		&{}+ \gam\left(\mv^\top+\left(\frac{1}{\gam}-1\right)\muv^\top C\right)\gv
		-\frac{\gam}{2}\tr(C^\top\Sig\,C\, H)
		=0,
	\end{align}
	for all $(t,\zv)\in[0,T]\times\Rb^N$, in which we have omitted the $t$ arguments to simplify the notation. Taking the terminal condition \eqref{eq:TC} into account, it then follows that $H$, $\gv$, and $f$ must satisfy \eqref{eq:H-ODE}, \eqref{eq:g-ODE}, and \eqref{eq:f-integral}, respectively.
	
	The verification result and the optimal trading strategy are obtained in a similar fashion as in the proof of Theorem \ref{thm:Futures}.
\end{document}